\documentclass[11pt]{article}  
\usepackage[square,numbers,sort&compress]{natbib}
\usepackage{doi}

\usepackage{hyperref}
\usepackage{authblk}
\usepackage{epsfig}  
\usepackage{amsmath}
\usepackage{braket}
\usepackage{amsfonts}
\usepackage{xcolor}
\usepackage{amsthm}
\usepackage{ bbold }
\usepackage{enumitem}   
\usepackage{amssymb}
\usepackage{multicol}
\usepackage{amstext} 
\usepackage{array}   
\usepackage{tabu}
\usepackage{framed}
\usepackage{subfig}
\usepackage{float}
\usepackage{cancel}
\usepackage{tikz-cd,mathtools}
\usepackage{algorithm}
\usepackage{algpseudocode}
\usepackage{latexsym}
\usepackage{wasysym}
\usepackage{tikz}
\usepackage{pgfplots}
\usetikzlibrary{shapes}
\usetikzlibrary{3d,perspective}

\newtheorem{theorem}{Theorem}[section]
\newtheorem{proposition}{Proposition}[section]

\newtheorem{problem}{Problem}[section]
\newtheorem{definition}{Definition}[section]
\newtheorem{corollary}{Corollary}[section]

\newtheorem{condition}{Condition}[section]
\newtheorem{example}{Example}[section]

\newcommand{\NL}{\mathsf{NL}}
\newcommand{\LL}{\mathsf{L}}
\newcommand{\x}{\mathbf{x}}

\newcommand{\order}[1]{\mathcal{O} \Big( #1 \Big)}

\newcommand{\norm}[1]{\lVert #1 \rVert }

\def\precdot{\mathrel{\prec\kern-.5em\hbox{$\cdot$}}}
\def\succdot{\mathrel{\hbox{$\cdot$}\kern-.5em\succ}}
\def\rddots{\cdot^{\cdot^{\cdot}}}

\DeclareMathSymbol{\shortminus}{\mathbin}{AMSa}{"39}

\DeclareMathOperator\rank{rank}

\newcommand{\mycircle}[5] 
{%
  \pgfmathsetmacro\vtheta{atan2(#2,#1)}                     
  \pgfmathsetmacro\vphi  {acos(#3/sqrt(#1*#1+#2*#2+#3*#3))} 
  \begin{scope}[rotate around z=\vtheta,rotate around y=\vphi,canvas is xy plane at z=#5]
    \draw[fill=gray,fill opacity=0.2] (0,0) circle (#4);   
  \end{scope}
}

\begin{document}  

\title{A Framework for Spatial Quantum Sensing}   

\author[1,2,3,4]{Lu\'{i}s Bugalho}
\author[1,2,3]{Yasser Omar}
\author[4]{Damian Markham}

\affil[1]{\footnotesize Instituto Superior Técnico, Universidade de Lisboa, Portugal}
\affil[2]{Physics of Information and Quantum Technologies Group, Centro de Física e\par Engenharia de Materiais Avançados (CeFEMA), Portugal}
\affil[3]{PQI -- Portuguese Quantum Institute, Portugal}
\affil[4]{Sorbonne Université, CNRS, LIP6, 4 Place Jussieu, Paris F-75005, France}
\date{}
\maketitle

\abstract{Quantum sensor networks promise precision advantages over classical and single-sensor strategies, in particular when the estimator is non-local. We address the problem of finding such estimators through a framework we connote spatial quantum sensing: given an underlying field interrogated by a network of quantum sensors at fixed positions, construct an estimator for a property of the field, for example, distinguishing a source of signal, or evaluating the field or its derivatives at an arbitrary point. We first treat polynomial fields, casting the task as an interpolation problem, and then generalize to fields modeled by analytic functions, which yields general least-squares estimators. A central and largely unaddressed question is under what conditions on sensor placement these estimators are well-defined and error-free. For $m$-dimensional arrays we give explicit constructions and proofs in the interpolation setting using algebraic geometry, and establish necessary and sufficient conditions in the general case. Comparing a non-local entangled protocol with the best local strategy, we show that entanglement yields maximal precision in distributed sensing under global resource constraints. Finally, we introduce error-free subspaces, a technique that translates prior knowledge of the field into a reduction in the number of required sensors. We expect these techniques to be broadly useful in sensing problems across scales, ranging from earth-scale experiments to local applications such as biological imaging.
}

\section{Introduction}

Distributed sensing and, in particular, distributed quantum sensing are good candidates for near-term quantum technologies, achieving results that lie beyond what is classically possible \cite{Paris2009,Komar2016,Riehle2017,Eldredge2018,Ge2018,Zhuang2018,Proctor2018d,Qian2019,Sidhu2019,Oh2020,Rubio2020a,Goldberg2021,Liu2021e}, in some scenarios. This work is motivated by early work \cite{Sekatski2019,Rubio2020a,Qian2020,Hamann2022,Hamann2023,Hamann2024}, including experimental realizations \cite{Bate2025}, where the idea of finding linear functions that allow the direct measurements of properties of a field was introduced and developed. Among these problems  one can find distinguishing the noise signals from the target signals \cite{Sekatski2019,Rubio2020a,Qian2020,Hamann2022}, or finding derivatives of the field at different locations \cite{Sekatski2019,Rubio2020a,Qian2019,Qian2020,Bringewatt2021,Hamann2022}. Our work extends this even further by building a framework that covers all the problems in these past works, and takes a step forward to new problems, while shining a light over the limiting aspects of such methodology and how to deal with them.

A key feature of these studies is the distributed nature of the sensing scenario, where the objective is to estimate a target function that is linear in the space of local parameters. This essentially means dealing with the problems for which one can build a linear estimator for a property of the field with access to the values of the field in a set of points in the domain. It has been shown that there exist not only quantum states capable of extracting only the desired linear function \cite{Bugalho2025}, but also that entangled states maximize the information that can be obtained about these functions \cite{Eldredge2018,Proctor2018d,Sekatski2019,Rubio2020a,Bugalho2025}.

This motivates the results of this paper, which focus on finding the categories of sensing problems for which one can find a linear estimator, where the above results imply that entangling quantum sensors gives an (optimal) advantage. In particular, we explore a description -- spatial sensing --, with multiple possible applications, where the parameters are given by interrogation of a global field with a quantum sensor. This field, as will be explained shortly, is simply a function that attributes a value to every point of the domain. The domain can be perceived as a spatial domain, with one, two or three dimensions, or even extend to the time dimension, if one is capable of creating time-entangled states ($e.g.$ time-bin entanglement \cite{Versteegh2015,Zhang2018}). It can also be, for example, the space of phases in a interferometry experiment, with multiple sources for the overall phase. This can also be extended to consider data and arbitrary functions of that data, such as statistical information of sets of data and so on. 

For this reason, there are several problems that could fit inside of this formalism. The most obvious ones correspond to physical fields one is trying to extract information of, for example the gravitational field or a magnetic field. For these examples, one can already find quantum sensors such as atom interferometers \cite{Abend2019}, and nitrogen-vacancy (NV) centers \cite{Ajoy2017,Kuwahata2020}, respectively. More recently, an experiment using a similar formalism to distinguish different sources of magnetic fields \cite{Bate2025} has been realized experimentally using trapped ion sensors.
 
In our work we start with the simple problem of the finite differences method to calculate linear estimators of higher order derivatives of the field in one dimension. The fact that these estimators are linear in the local values of the field, accessible by the quantum sensors, is again what provides the quantum advantage in the estimation task \cite{Proctor2018d,Eldredge2018,Bugalho2025}. This provides a good intuition on how to read and unravel the forthcoming results. We use the analysis of this problem and generalize first for finite differences in higher dimensions, using tools from algebraic geometry. As the main object of study in algebraic geometry are polynomials, we use these tools to generalize our results to polynomial interpolation of functions, in the domain where the polynomial approximation is valid. This problem has an intersection with previous works \cite{Sekatski2019,Rubio2020a,Qian2019,Qian2020,Bringewatt2021,Hamann2022,Hamann2023,Hamann2024}, where the authors have also treated the potential use cases of their methods to address the problem of finding the Taylor expansion coefficients of the field. This is a problem which is equivalent to the problem of interpolation. We extend over these past works to reach the full potential of the information Taylor expansion can provide in a distributed scenario. In particular, we find a general way to create linear estimators for a set of problems dealing only with polynomial fields, namely, multi-dimensional higher order derivatives, interpolation values of the field, interpolation values of derivatives of the field and any linear combination thereof. Importantly, we also discuss and provide results for the error in the construction of the linear estimators in these problems, which is not obvious from the results in \cite{Sekatski2019,Hamann2022,Hamann2023,Hamann2024}. This provides more insight in these methods, guaranteeing that no errors arise from the information not discernible by the quantum sensors.

We then use this formalism for polynomial interpolation and extend it to analytical fields \cite{Qian2019,Qian2020,Bringewatt2021}, where the available functions (i.e. those used to decompose the field) are no longer just polynomial functions, but any analytical function. In doing so, we recover the \textit{decoherence free subspaces} from \cite{Sekatski2019,Hamann2022,Hamann2023}, where the problem was stated as signal isolation. Given that polynomials are analytical functions, the previous problem is contained as well in this new problem. In fact, we go a step even further, to what we believe ends the set of chained problems that can be express in this form, by considering least-squares methods. This expands on the idea of modeling the field with a set of functions, and facilitates the use of additional statistical tools, such as regularization techniques. In particular, we find a set of problems, chained in inclusions, meaning the previous problem is included in the following problem, by restricting the assumptions on the functions and properties of the problem:

\begin{equation}
	\mathsf{Interpolation} \subset \mathsf{Signal-Isolation} \subset \mathsf{Least-Squares}
\label{eq:inclusion}
\end{equation}

The structure of the paper is as follows: we start by detailing the general problem of field sensing in Section~\ref{section:field}, which is the main framework of this paper. Then we go through the problem of interpolation in Section~\ref{section:interpolation}, providing exact solutions with necessary and sufficient conditions for the existence of a solution of the problem. We then analyze the case for analytical fields, starting with the signal isolation problem in Section~\ref{section:signal} and finally leading to the least-squares method in Section~\ref{section:leastsquares}, providing as well necessary and sufficient conditions for the existence of a solution. This order comes from the order introduced above in Eq.~\ref{eq:inclusion}. Finally we discuss the meaning of the non-solvability of the problem in Section~\ref{section:inverting}, how to construct the problem in a way that by default creates \textit{error-free subspaces}. These subspaces depend on the variables of the sensors, \textit{e.g.} the positions, and allow one to find a solution of the problem that by construction does not involve an error propagated from what the model cannot discern. We also provide examples for each of the categories of these problems and compare the best quantum non-local strategy with the best quantum local strategies and the best classical strategies.

\section{Field Sensing}\label{section:field}

We start by defining exactly what we mean by spatial sensing. We will do it in a general way, so many problems can fit within this framework. The main idea lies in having a function that takes each point over a space, or a domain, and attributes them a value. The points in the domain are variables of each of the parties, such as their position. In principle these values can be unknown, such as private location data, but this goes beyond the scope of this work and would require alternative methods than the ones introduced here. Hence we will assume knowledge of the positions / points on the domain. Then, what each party has access to, or is able to encode in each of the sensors, is the value of the function, or field, at their point in the domain, or position. 

In most regular scenarios, the dimension of our problem is an integer $m$, and therefore we are working with a domain $D \subseteq \mathbb{K}^m$. Note that for the most common problems that fit this field sensing scenario, $\mathbb{K}$ is usually given by the reals $\mathbb{R}$, as it defines some spatial dimension, and $m$ ranges from 1 to 3.

The first component is what can be called the model of our data, a subset of analytical functions over the domain, $\mathcal{F} = \{ f_1, f_2, \cdots , f_k \} \subseteq C^\omega (D) $, where $C^\omega$ is the standard notation for analytical functions over some domain $D$. These functions model our field, in the sense that the field can potentially be written as a linear combination of the functions in $\mathcal{F}$:

\begin{equation}
\begin{aligned}
    \tilde{F} : D \subseteq \mathbb{K}^m &\longrightarrow \mathbb{R} \\
    x &\longmapsto F(x) 
\end{aligned} 
\ , \qquad \tilde{F}(x) = \sum_{j=1}^k \beta_j f_j(x), \qquad f_j \in \mathcal{F} ,
\label{eq:fielddefinition} 
\end{equation}
similar to the works of \cite{Sekatski2019,Qian2019,Qian2020,Bringewatt2021,Hamann2023,Hamann2024}. This encompasses many situations of interest,  namely a field with several sources of signal, as well as series expansions such as the Taylor expansion, where the functions $f_j$ are monomials, and the Fourier expansion, where the functions $f_j$ are sines and cosines with different frequencies. 

The second component is access to the field in different points of its domain where our sensors sit. 
For this we use a set $X = \{x_1, x_2, \cdots , x_p\}$ which is a discrete collection of points in $D$, $X\subset D$, with size $p$. Moreover, each party has access to a quantum channel that is able to interrogate or encode the value of the field at their point in the domain. This is the link to a quantum estimation scenario. This interrogation is done by applying the channel over a quantum state described by a density matrix $\rho$ that can be entangled over several sensors, which we call a non-local resource. In particular, each of these quantum channels can be described by some quantum dynamics, for example:
\begin{equation}
\begin{aligned}
	\Lambda_x (\rho) \equiv \Lambda_{F(x)} (\rho) = e^{i F(x) H_x} \rho e^{-i F(x) H_x} .
\end{aligned}
\end{equation}

\begin{figure}[t]
\centering
\subfloat{\includegraphics[width=0.75 \linewidth]{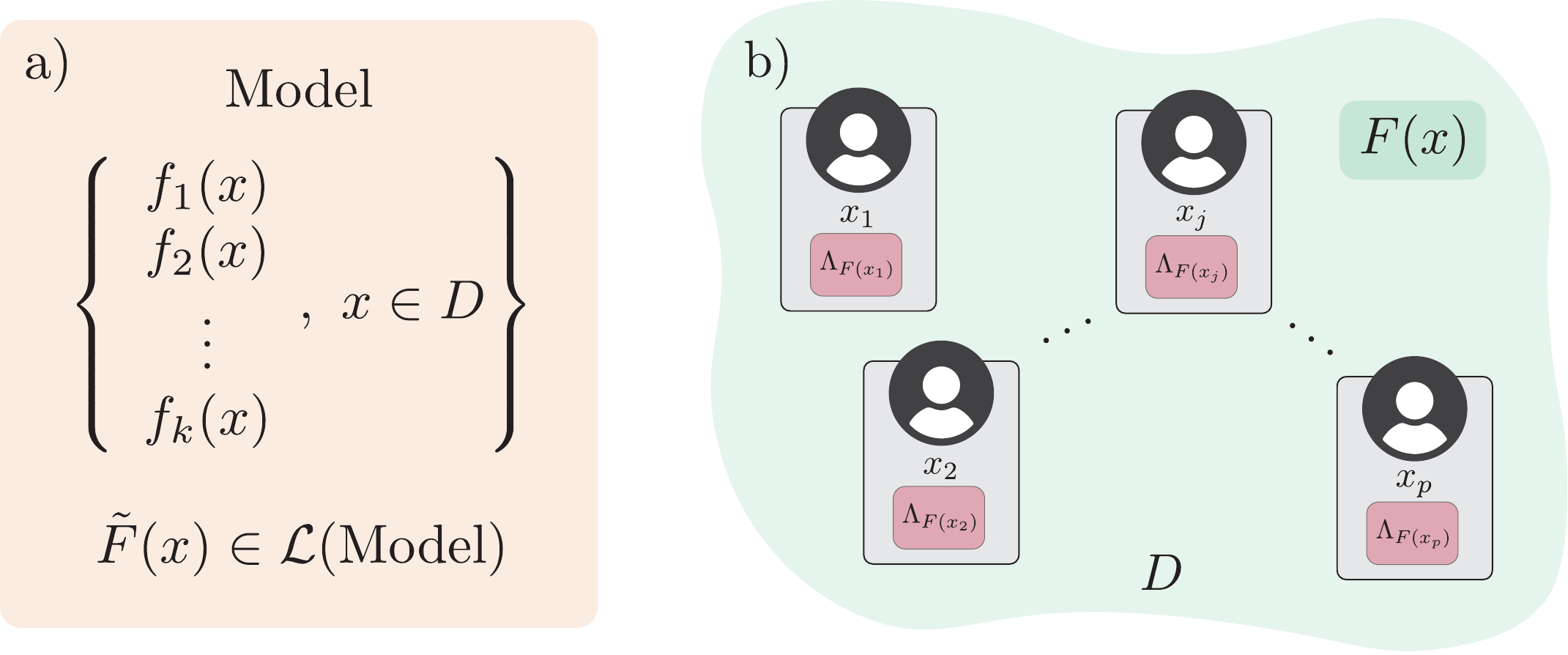}} 
\caption{Schematic representation of the different parts of the problem setup. \textit{a)} Mathematical model defined by a set of functions $ \{f_1(x), f_2(x), \dots, f_k(x)\} $ over a domain $ D $. The estimator for the field, or underlying function $\tilde{F}$ belongs to the linear span of the model functions $\mathcal{L}(\text{Model})$. \textit{b)} All parties have access to the underlying field $F(x)$ at their own points in the domain $ X = \{ x_1, x_2, \dots, x_j, \dots, x_p \} \in D$. In particular, they have access to some channel that encodes the outputs $ \{ \Lambda_{F(x_1)},\Lambda_{F(x_2)}, \dots, \Lambda_{F(x_j)}, \dots, \Lambda_{F(x_p)} \}$.}
\label{fig:scheme}
\end{figure}  

Under these assumptions over the quantum dynamics, it is known that an entangled state with access to each of the local channels, and control over a coupling constant of each the Hamiltonian (equivalent to mapping $H_x \rightarrow c_x H$) will be optimal at estimating the linear function $\boldsymbol{c}\cdot \boldsymbol{F} = \sum_j c_{x_j}F(x_j) \equiv \sum_j c_{j}F(x_j)$ \cite{Eldredge2018,Proctor2018d,Sekatski2019,Rubio2020a,Bugalho2025}. This means that, if one finds a linear function over the values of the field for a target property of the field, then one is able to find an appropriate initial quantum state (generally engantled) maximizing the precision of estimating the target property of the function. The advantage comes in terms of precision of the estimator for $\boldsymbol{c}\cdot \boldsymbol{F}$, and depends on the vector $ \boldsymbol{c}$. In particular, the gain in precision is lower-bounded by $\sqrt{n}$ when $\boldsymbol{c}$ has only one non-zero entry (local linear function), using $n$ local resources. On the other hand, the upper bound in precision is extended to $\sqrt{n \cdot p}$, again for $n$ local resources and $p$ entangled sensors, when $\boldsymbol{c} \propto \vec{1}$. This recovers the Heisenberg scaling on the total number of resources $N=n\cdot p$. The exact form of the advantage depends on the constraints of the resources, and the values of $\boldsymbol{c}$, and is thus different case by case. Nevertheless it can always be found using our framework developed here. We shine a light on this in one of the examples of this work (see Section~\ref{section:interpolationexample}), but leave a more detailed analysis for future work. 

To this end, the goal then becomes finding the set of problems, and how to build them, such that 
 they can be viewed as the problem of estimating a linear function of the field  at the points of the detectors $\boldsymbol{c}\cdot \boldsymbol{F}$. The details of the problem and property of interest will define the linear function $\boldsymbol{c}$, as well as give us conditions on the position and number of sensors that are needed. This will depend on the assumptions that we have about the functions $\mathcal{F}$ as well as the property that we are interested in, and the constraints on the sensor locations and resources. The approach and tools we put to this have foundations in analytical and algebraic geometry. 

The main intuition behind every case we will discuss, wether it be the interpolation problem, the signal-isolation or the least-squares, is finding the closest function in the linear span of our model, that best describes the observed values. The solution of this problem can always be cast as the solution to a linear problem as follows:

\begin{equation} \label{eqn global problem}
\boldsymbol{X} \boldsymbol{\beta}= \boldsymbol{F} \Leftrightarrow 
\begin{pmatrix}
f_1(x_1) & f_2(x_1) & \cdots & f_k(x_1) \\
f_1(x_2) & f_2(x_2) & \cdots & f_k(x_2) \\
\vdots   & \vdots   & \ddots & \vdots   \\
f_1(x_p) & f_2(x_p) & \cdots & f_k(x_p)
\end{pmatrix} 
\begin{pmatrix}
\beta_1 \\
\beta_2 \\
\vdots \\
\beta_k
\end{pmatrix} =
\begin{pmatrix}
F(x_1) \\
F(x_2) \\
\vdots \\
F(x_p)
\end{pmatrix}
\end{equation}
where $\boldsymbol{X}$ is the matrix in the above equation, $\boldsymbol{\beta}$ is the vector of the model coefficients, and $\boldsymbol{F}$ is the vector of the evaluations of the field, or the local values of the field.

By inverting this system we have an estimator for the vector $\hat{\boldsymbol{\beta}}$, which implicitly defines the closest function $\tilde{F}(x)$ to $F(x)$ that lies inside of the linear span of our model. Here, the notion of closeness is given by the minimization of the squared deviation between the local data and the expected data from the model. This means that, if one is able to write each $\beta_i = \boldsymbol{c}_i \cdot \boldsymbol{F}$, then finding the value of the field at some point $\tilde{F}(x)$ amounts to estimating 
\begin{equation}
\begin{aligned}
    \tilde{F}(x) &= \sum_i \beta_i f_i(x)= \sum_i ( \boldsymbol{c}_i \cdot \boldsymbol{F} ) f_i(x) = \sum_i ( f_i(x) \boldsymbol{c}_i ) \cdot \boldsymbol{F} = \boldsymbol{c} \cdot \boldsymbol{F},
\end{aligned}
\end{equation}
which is exactly the form we want, where we know the optimal quantum strategy and how entanglement helps. Similarly one can form an expression for the derivative at any point, e.t.c. 
In this way, the problem reduces to finding a solution to equation (\ref{eqn global problem}). The form or constraints on when this can be done is one of the main questions addressed in this work, and it depends on the assumptions on the set of functions in our model $\mathcal{F}$ and the position and number of sensors. 
If the values of the field have some statistical error, this can also be incorporated by defining a pseudo-inverse taking this into account. Depending on the different assumptions of how we build our model and how many data points we have, we will build the different problems and we will be able to provide statements over what the model is able to discern or not, depending on the positions of the sensors.

We will start by analyzing the simplest case, which is that of interpolation, in the upcoming section. One can also think of the interpolation problem as finding the Taylor expansion coefficients of the field in question, which allows us to deal with numerical derivatives or linear functions of the numerical derivatives of the field. This is a consequence of interpolation being an estimation technique that deals only with polynomials, a subset of analytical functions that constitute the core of algebraic geometry.
From there in later sections we address the problem when we know the functions are from a set of more general algebraic functions, and finally least squares.

\section{Interpolation}\label{section:interpolation}

In this section we will cover the interpolation problem, which is the simplest to analyze, and where we can find concrete tools that grant a well-defined model, and its consequences. At the same time, this is a useful tool for when little to no \textit{a-priori} information about the field is known, besides that it is an analytical function. If no other information is known, then there should be no privileged set of functions for the model over another. Therefore, in this section the goal is given access to the value of the field in a set of points, find a polynomial interpolation of the field. 

By taking the Taylor expansion we are choosing a model where the functions are monomials in increasing exponents order, which is something that is available from the definition of an analytical function: it admits a Taylor expansion around a neighborhood of the points at which the function is analytical. When a field is a polynomial, finding the interpolating function using $p$ points is the same problem as finding the Taylor expansion of the function up to $p$-orders of expansion. 

In this section we will start by addressing how to properly do a Taylor expansion in multiple dimensions. Then we use this to calculate the expansion coefficients from the local values accessible by the sensors, which amounts to the inversion of a generalized Vandermonde matrix. In doing so, we will be able to write the interpolation of the field in arbitrary positions (as well as derivatives of the interpolation of the field) as a linear function of the local values. Additionally, we discuss the impact of the positions of the sensors in the existence of an error-free solution to the problem. This grants the existence of a linear estimator for the interpolation problem.

\subsection{Taylor Expansion}
In one dimension, this is a simple problem; however, for multiple dimensions, the question of how to find the increasing orders arises. Starting with a simple example in 1D, the Taylor expansion of $F(x)$ around some point $x_0$ up to order $k$ is given by:

\begin{equation}
	F(x)=\sum_{j=0}^{k-1} \frac{ F^{(j)} (x_0) }{j!} \delta x^j + \mathcal{O}\left(\delta x^k \right), \delta x = x -x_0.
	\label{eq:expansion1D}
\end{equation}
In this case, the order is well defined, as the exponents are given by the integers, which have a total order $\leq$, meaning we can put them in a sequence of increasing numbers. When considering domain with higher dimensions, this is no longer true. A consequence of this is that the order of the errors, or the order of the expansion, can have several options. This is linked with the fact that the product order $\preceq$ is known not to be a total order. To deal with this we use a concept from order theory that will help define a proper a multivariate Taylor expansion. The multivariate derivatives in $m$-dimensions can be labeled by an $m$-length vector of integers $\boldsymbol{\alpha}$, where the component $j$ corresponds to the degree of the derivative in the dimension $x_j$. In doing so, let us use the common notation $D^{\boldsymbol{\alpha}}$, where $\boldsymbol{\alpha} \in \mathbb{N}_0^m$, to denote the derivative $\partial_{x_{1}}^{\alpha_1} \circ \partial_{x_{2}}^{\alpha_2} \circ \cdots \circ \partial_{x_{m}}^{\alpha_m}$. When doing a Taylor expansion in any dimension, a proper one guarantees that lower order derivatives are accounted for before increasing the degree in any direction. Because of this we will use the some definitions to help define a proper multivariate Taylor expansion:

\begin{definition}[Lower Set]
Let $(P,\preceq)$ be a partially ordered set. A lower set is a subset $I \subseteq P$ such that:
\begin{enumerate}
	\item $I$ is non-empty
	\item $\forall \ x \in I , y \in P$, $y \preceq x$ implies $y \in I$
\end{enumerate}
\label{def:lowerset}
\end{definition}
Additionally define:
\begin{definition}[Border and Cover]
Let $L$ be a lower set in a partially ordered set $(P,\preceq)$ and $P_L$ the subset of $P$ defined by excluding $L$, $P_L = P\setminus L$.
\begin{enumerate}
	\item Let $\partial L \subseteq L$ be the border of $L$ : $\partial L = \{ x \in L \ | \not\exists z \in L\setminus \{x\} : x \prec z \}$. 
	\item Let $\mathsf{Cov}(L)$ be the cover of $L$ : $\mathsf{Cov}(L)  =  \{ x \in P_L \ | \ \exists z \in L :  z \precdot x  \}$.
\end{enumerate}
\end{definition}

\begin{example}\label{ex:lowerset}
We provide three examples of lower sets of $\mathbb{N}_0^2$ in black, and highlight in orange the cover of the lower sets. The points surrounded by gray belong to the border of the lower set as well.
\end{example}

\vspace{0.2cm}

\begin{tikzpicture}[scale=0.55]
\draw[help lines, color=gray!30, dashed] (0,0) grid (5.9,5.9);
\foreach \x in {0,1,2,3,4 }{
   \foreach \y in {0,1,2,3,4}{
	\draw[black,fill=black] (\x,\y) circle (.75ex);
    }
}
\foreach \Point in {(0,5),(1,5),(2,5),(3,5),(4,5),(5,4),(5,3),(5,2),(5,1),(5,0)}{
    \draw[orange,fill=orange] \Point circle (.75ex);
}
\draw[gray] (4,4) circle (1.25ex);
\node at (8,3) {$$};
\end{tikzpicture}
\begin{tikzpicture}[scale=0.55]
\draw[help lines, color=gray!30, dashed] (0,0) grid (5.9,5.9);
\foreach \x in {0,1,2,3}{
   \foreach \y in {0,1,2,3}{
    	\draw[black,fill=black] (\x,\y) circle (.75ex);
    }
}
\foreach \Point in {(1,1), (1,2), (1,3), (1,4), (0,4), (2,1), (2,2), (2,3) ,(3,1), (3,2), (3,3), (4,1),(4,0)}{
   \draw[black,fill=black] \Point circle (.75ex);
}
\foreach \Point in {(0,5), (1,5), (2,4), (3,4), (4,3),(4,2),(5,1),(5,0)}{
    \draw[orange,fill=orange] \Point circle (.75ex);
}
\draw[gray] (4,1) circle (1.25ex);
\draw[gray] (3,3) circle (1.25ex);
\draw[gray] (1,4) circle (1.25ex);
\node at (8,3) {$$};
\end{tikzpicture}
\begin{tikzpicture}[scale=0.55]
\draw[help lines, color=gray!30, dashed] (0,0) grid (5.9,5.9);
\foreach \x in {4,...,0}{
   \foreach \y in {\x,...,0}{
    	\draw[black,fill=black] (4-\x,\y) circle (.75ex);
    }
}
\foreach \x in {5,...,0}{
    	\draw[orange,fill=orange] (\x,5-\x) circle (.75ex);
}
\foreach \x in {4,...,0}{
    	\draw[gray] (\x,4-\x) circle (1.25ex);
}
\end{tikzpicture}

\vspace{0.5cm}

Note that, if we are working with vectors of integers or reals, equipped with the regular product order, the border of any lower set can be seen as the Pareto front of the subset, this is, the set of non-dominating vectors. Using this let us write a proper Taylor expansion for the multivariate case, given a lower set $L \subset \mathbb{N}_0^m$:

\begin{equation}
	F(\boldsymbol{x})=\sum_{\boldsymbol{\alpha} \in L} \frac{D^{\boldsymbol{\alpha}} F(\boldsymbol{x}_0)}{\boldsymbol{\alpha} !}\delta\boldsymbol{x}^{\boldsymbol{\alpha}}+\sum_{\boldsymbol{\alpha}\in \mathsf{Cov}(L)} h_{\boldsymbol{\alpha}}(\boldsymbol{x}) \delta\boldsymbol{x}^{\boldsymbol{\alpha}},
\label{eq:taylormulti}
\end{equation}
where $\delta \boldsymbol{x} = \boldsymbol{x} - \boldsymbol{x}_0$ and using the notation $\boldsymbol{\alpha} = (\alpha_1, \alpha_2, \cdots , \alpha_m)$, $\boldsymbol{\alpha}!=\alpha_1!\alpha_2!\cdots \alpha_m!$ and $\boldsymbol{y}^{\boldsymbol{\alpha}} = y_1^{\alpha_1} y_2^{\alpha_2} \cdots y_n^{\alpha_n}$. Note that in Example \ref{ex:lowerset}, the leftmost and rightmost lower sets are equivalent to imposing $\norm{\boldsymbol{\alpha}}_\infty \leq k$ and $\norm{\boldsymbol{\alpha}}_1 \leq k$, respectively, for some $k$ given by the number of points.
We exemplify the Taylor expansion for $m=2$ using a lower set similar to the rightmost of Example \ref{ex:lowerset}, but for $\norm{\boldsymbol{\alpha}}_1 \leq 2$:
\begin{equation}
\begin{aligned}
	F(x,y) &\approx F(x_0,y_0) + \frac{\partial F(x_0,y_0)}{\partial x} (x-x_0) + \frac{\partial F(x_0,y_0)}{\partial y} (y-y_0) + \\
	 &\frac{\partial^2 F(x_0,y_0)}{\partial x^2} \frac{(x-x_0)^2}{2!} + \frac{\partial^2 F(x_0,y_0)}{\partial y^2} \frac{(y-y_0)^2}{2!} + \frac{\partial^2 F(x_0,y_0)}{\partial x \partial y} (x-x_0)(y-y_0) +\cdots
\end{aligned}
\end{equation}

If one were to omit a term such as a first order derivative in the above example, then the expansion would have an error in a lower order. This is the reason behind using a lower set for the orders of the expansion. Next we discuss how to find the derivatives of the function from the Taylor expansion.

\subsubsection{Expanding to the Derivatives of a Function}\label{section:derivatives}

While Taylor theorem tells us that one can find the approximation up to some order of a function at a point by looking at its derivatives in a point close enough, and constructing a polynomial approximation of this, the polynomial approximation also admits a derivative up to a certain degree. So, if one is interested in calculating the approximations of a derivative at a point, one can simply do the following:

\begin{equation}
	D^{\zeta} F(\boldsymbol{x}) =\sum_{\boldsymbol{\alpha} \in L, \boldsymbol{\zeta} \preceq \boldsymbol{\alpha} } \frac{D^{\boldsymbol{\alpha}} F(\boldsymbol{x}_0)}{(\boldsymbol{\alpha}-\boldsymbol{\zeta}) !}\delta\boldsymbol{x}^{\boldsymbol{\alpha}-\boldsymbol{\zeta}} + \order{ \sum_{\boldsymbol{\alpha} \in \mathsf{Cov(L)}}  \delta\boldsymbol{x}^{\boldsymbol{\alpha}-\boldsymbol{\zeta}}}
\label{eq:expansionderivatives}
\end{equation}

Note that this is very similar to Eq.\ref{eq:taylormulti}, without summing some orders of expansion, and with different coefficients weighting each derivative. This results from acting $D^{\boldsymbol{\eta}}$ on Eq.~\ref{eq:taylormulti}.

\subsection{Transforming the Expansion in a Linear Problem}

As noted before the quantum sensors we will have access to the field at a given collection of points. Let us collect these values on a vector as follows: $\boldsymbol{F} = \{F(x_j)\}_{x_j \in X}$. The goal in this section is to transform the problem of calculating first a derivative of the field $F(x)$ at a given location of a sensor $x_j$, and then generalize to evaluating the field or a derivative of the field outside of $X$ by interpolation, and then generalize even more to any linear combination thereof. The last part comes naturally since linear combinations of linear combinations are still linear.

Finding the derivative at a given $x_j$, using a set of evaluations of a function at different points can be posed as a numerical differentiation problem. To deal with this problem one can use, for example, a finite differences method. For simplicity let us start with the 1D case of a set of $p$ points across one line and expand the field in each location around $x_j$, using $\delta x_{ij}$ to denote the difference $x_i - x_j$:
\begin{equation}
\begin{aligned}
	F(x_1) = F(x_j + \delta x_{1j}) \approx \ & F(x_j) +  F'(x_j) \frac{\delta x_{1j}}{1!} + F''(x_j) \frac{\delta x_{1j}^2}{2!}  + ... \\
	F(x_2) = F(x_j+ \delta x_{2j}) \approx \ & F(x_j) +  F'(x_j) \frac{\delta x_{2j}}{1!} + F''(x_j) \frac{\delta x_{2j}^2}{2!} + ... \\
	&\cdots \\
	F(x_p)= F(x_j+  \delta x_{pj}) \approx \ & F(x_j) +  F'(x_j) \frac{\delta x_{pj}}{1!} + F''(x_j) \frac{\delta x_{pj}^2}{2!} + ... 	
\end{aligned}
\label{eq:expansions}
\end{equation}
Then, the idea is inverting this problem, by finding a linear combination of the $F(x_i)$ such that:

\begin{equation}
	\boldsymbol{c}\cdot \boldsymbol{F} = c_1 F(x_1) + c_2 F(x_2) + \cdots + c_p F(x_p) \approx F^{(r)} (x_j)  
	\label{eq:findingDerivative}
\end{equation}
For some $0\leq r < p$. The coefficients $c_j$ can be calculated by substituting Eq. ~\ref{eq:expansions} into Eq.~\ref{eq:findingDerivative}. After this substitution, finding the coefficients is equivalent to solving a linear system with a matrix, which is often denoted as the Vandermonde matrix:

\begin{equation}
\begin{aligned}
	\begin{pmatrix}
	1 & 1 & \cdots & 1 \\
	\delta x_{1j} & \delta x_{2j} & \cdots & \delta x_{pj} \\
	\vdots & \vdots & & \vdots \\
	\delta x_{1j}^r & \delta x_{2j}^r & \cdots & \delta x_{pj}^r \\ 
	\vdots & \vdots & & \vdots \\
	\delta x_{1j}^{p-1} & \delta x_{2j}^{p-1} & \cdots & \delta x_{pj}^{p-1} 
	\end{pmatrix}
	\begin{pmatrix}
	c_1 \\
	c_2 \\
	\vdots \\
	c_r \\
	\vdots \\
	c_p
	\end{pmatrix}
	=
	\begin{pmatrix}
	0 \\
	0 \\
	\vdots \\
	r! \\
	\vdots \\
	0
	\end{pmatrix}
\end{aligned}
\Leftrightarrow V(\delta_j X)^T \boldsymbol{c} = r! \boldsymbol{e}_r 
\label{eq:vandermonde1D}
\end{equation}
Where $V(\delta_j X)$ is the Vandermonde matrix taking data from a set of points $\delta_j X = \{ x-x_j\}_{x\in X}$. Note that this matrix is square and invertible if and only if every $\delta x_{ij}$ is different, meaning always choosing different points for the location of sensors. Later one we will provide a proof for this and the more general case of a multivariate expansion. Moreover, this will imply an error of the approximation of the derivative in the order of $\order{\delta x_{ij}^p}$, which is expected from Eq.~\ref{eq:expansion1D}.

The case of multiple dimensions, the problem is more complex. While there seems to be plenty of literature for finite differences methods in one dimension, for more than one dimension, an equivalent solution using a generalized Vandermonde matrix seems not to be commonly analyzed in the literature, even though some early work shines a light on this extension \cite{Buck1992,Lundengard2017}. We believe there are two reasons for this: defining a proper expansion and the conditions for invertibility of the generalized Vandermonde matrix. The first problem, we have already addressed by introducing lower sets. The second problem we will solve using algebraic geometry, which will provide some insight to deal with the more general problem of any analytical function.

To build the generalized Vandermonde matrix in $m$ dimensions, take a lower set $L = \{ \boldsymbol{\alpha}_1, \boldsymbol{\alpha}_2, \cdots , \boldsymbol{\alpha}_p\} \subseteq \mathbb{N}_0^m$ of size $p$ for the expansion derivatives coefficients, and a set of points $X =  \{ \boldsymbol{x}_1, \boldsymbol{x}_2, \cdots , \boldsymbol{x}_p\} \subset \mathbb{K}^m$ with the same size $p$. In the same way as before, construct the vector $\boldsymbol{F}$ by evaluating the field $F$ at each of the points in $X$, $\boldsymbol{F}=\{ F(\boldsymbol{x})  \}_{\boldsymbol{x}\in X}$. Expanding each of them around $\boldsymbol{x}_j$, using again the notation $\delta \boldsymbol{x}_{ij} = \boldsymbol{x}_{i} - \boldsymbol{x}_{j} $:

\begin{equation}
\begin{aligned}
	F(\boldsymbol{x}_1) = F(\boldsymbol{x}_j + \delta \boldsymbol{x}_{1j}) \approx \ &  \sum_{\boldsymbol{\alpha} \in L} D^{\boldsymbol{\alpha}}F(\boldsymbol{x}_j) \frac{(\delta \boldsymbol{x}_{1j})^{\boldsymbol{\alpha}}}{\boldsymbol{\alpha}!} + \epsilon (\boldsymbol{\alpha},  \boldsymbol{x}_1,\boldsymbol{x}_j) \\
	F(\boldsymbol{x}_2) = F(\boldsymbol{x}_j + \delta \boldsymbol{x}_{2j}) \approx \ &  \sum_{\boldsymbol{\alpha} \in L} D^{\boldsymbol{\alpha}}F(\boldsymbol{x}_j) \frac{(\delta \boldsymbol{x}_{2j})^{\boldsymbol{\alpha}}}{\boldsymbol{\alpha}!} + \epsilon (\boldsymbol{\alpha},  \boldsymbol{x}_2,\boldsymbol{x}_j) \\
	&\cdots \\
	F(\boldsymbol{x}_p) = F(\boldsymbol{x}_j + \delta \boldsymbol{x}_{pj}) \approx \ & \sum_{\boldsymbol{\alpha} \in L} D^{\boldsymbol{\alpha}}F(\boldsymbol{x}_j) \frac{(\delta \boldsymbol{x}_{pj})^{\boldsymbol{\alpha}}}{\boldsymbol{\alpha}!} + \epsilon (\boldsymbol{\alpha},  \boldsymbol{x}_p,\boldsymbol{x}_j) 
\end{aligned}
\label{eq:expansionsMulti}
\end{equation}
The problem is again finding the correct linear combination of $F(\boldsymbol{x}_i)$ such that:
\begin{equation}
	\boldsymbol{c}\cdot \boldsymbol{F} = c_1 F(\boldsymbol{x}_1) + c_2 F(\boldsymbol{x}_2) + \cdots + c_p F(\boldsymbol{x}_p) \approx  D^{\boldsymbol{r}} F (\boldsymbol{x}_j) 
\label{eq:taylormulticoeff}
\end{equation}
For some $\boldsymbol{r} \in L$. Using the same notation as before $\boldsymbol{x}^{\boldsymbol{\alpha}} = x_1^{\alpha_1} x_2^{\alpha_2} \cdots x_m^{\alpha_m}$ and $\delta_j X = \{\boldsymbol{x} - \boldsymbol{x}_j\}_{x\in X}$, one arrives at the following system of equations:
\begin{equation}
\begin{aligned}
	V(\delta_j X,L)^T \boldsymbol{c} =  \boldsymbol{r}! \boldsymbol{e}_{\boldsymbol{r}}, \qquad \big( V(\delta_j X,L) \big)_{ik} = \delta \boldsymbol{x}_{ij}^{\boldsymbol{\alpha}_k} = ( \boldsymbol{x}_{i} - \boldsymbol{x}_{j} )^{\boldsymbol{\alpha}_k}
\end{aligned}
\label{eq:vandermonde}
\end{equation}

To grasp this equation a little better, let us present an example:

\begin{example}
Let $X$ and $L$ be given by:
\begin{equation}
X = \left(\begin{array}{ccccc}
x_0 & x_1 & x_2 & x_3 & x_4 \\
y_0 & y_1 & y_2 & y_3 & y_4
\end{array}\right) \quad, \quad
L = \left(\begin{array}{ccccc}
0 & 1 & 2 & 0 & 1 \\
0 & 0 & 0 & 1 & 1
\end{array}\right) \equiv
\begin{tikzpicture}[scale=0.5,baseline={([yshift=-.5ex]current bounding box.center)},vertex/.style={anchor=base,
    circle,fill=black!25,minimum size=18pt,inner sep=2pt}]
\draw[help lines, color=gray!30, dashed] (0,0) grid (2.9,1.7);
\foreach \Point in {(0,0),(1,0),(2,0),(0,1),(1,1)}{
	 \draw[black,fill=black] \Point circle (.75ex);
}
\node [black] at (0,-0.3) {};
\end{tikzpicture}
\end{equation}
Then, the generalized Vandermonde matrix of this system is given by:
\begin{equation}
V(X,L)=
\left(\begin{array}{ccccc}
1 & x_0 & x_0^2 & y_0 & x_0 y_0 \\
1 & x_1 & x_1^2 & y_1 & x_1 y_1 \\
1 & x_2 & x_2^2 & y_2 & x_2 y_2 \\
1 & x_3 & x_3^2 & y_3 & x_3 y_3 \\
1 & x_4 & x_4^2 & y_4 & x_4 y_4  
\end{array}\right)
\end{equation}
This is equivalent to choosing the functions in the model to be given by $\{1,x,x^2,y,xy\}$.
\end{example}

As we have said before, for 1 dimension this matrix was a regular Vandermonde matrix, and we were granted invertibility simply from choosing all of the points to be different, $i.e.$ choosing different positions for our sensors. In this case, this condition is still necessary (to check this just see what happened to the matrix above if two points are equal), but it is no longer sufficient. One can rewrite the conditions of invertibility of the generalized Vandermonde matrix in an algebraic geometry problem, and use appropriate tools to study the invertibility. Let us make some statements, starting by an informal theorem, and then a formal condition with proof of sufficiency for the invertibility of the generalized Vandermonde matrix:

\begin{theorem}[Informal]\label{thm:informal}
Let $V(X,L)$ be a generalized Vandermonde matrix with $X$ a set of points in a $m$-dimensional space, and $L$ a set of different monomials in $m$ variables. $V(X,L)$ is invertible if and only if no polynomial spanned by $L$ vanishes for all $x\in X$.
\end{theorem}

We prove a more formal statement and do a brief introduction of important concepts required for the proof in Appendices \ref{appendix:vandermonde} and \ref{appendix:algebraic}, respectively. Using these results, we can already provide a statement for a certain type of sensor positions resulting in an invertible Vandermonde matrix, and, therefore, a well defined model:

\begin{proposition}\label{prop:invertibilityX}
Let $X \subset \mathbb{N}_0^m$ be a lower set for some arbitrary dimension $m$. $V(X,X)$ is invertible.
\end{proposition}
\begin{proof}
The proof comes directly from constructing a set of ideals in the polynomial ring, in the same way as Example~\ref{example:borderideals}. Then one can verify Condition~\ref{condition:invertibility1} from the degrees of the ideals of the set $X$ only taking values from the cover $\mathsf{Cov}(X)$. This means that the minimum degree polynomial belonging to $I(X)$ has degree in $\mathsf{Cov}(X)$, and not in $X$.
\end{proof}

The above theorem says that positioning sensors in the lower set of a grid, grants invertibility for the corresponding lower set in $\mathbb{N}_0^m$. This means I can estimate the expansion coefficients without an inherent error arising from the non-invertibility of the Vandermonde matrix.

Unsurprisingly, there exists more possible set of points $X$ that still result in an invertible Vandermonde matrix for a given lower set $L$. This corresponds to finding positioning for the sensors that allow to evaluate the coefficients of the expansion without any error.
From the proofs presented in Appendices \ref{appendix:vandermonde} and \ref{appendix:algebraic} one can observe certain invariances which can be exploited to generalize our results. The first one corresponds to global contractions and translations of all the points, which should not affect the invertibility. Furthermore, looking at Example~\ref{example:borderideals} in Appendix~\ref{appendix:vandermonde}, we have defined the points arbitrarily, which actually means we do not need a notion of order on the points. This suggests that the set of points $X$ (sensor locations) needs to be transformable into a lower set $L$, via changing lines and columns. An example of these transformations:

\begin{example}[Transforming a set $X$ into a lower set $L$ in two dimensions]\label{example:paretoset}
Let $c_i$ and $r_i$ stand for row $i$ and column $i$. We first contract the set around the center point. Then, doing the following set of permutations we get the third transformation: $c_3 \leftrightarrow c_1, l_3 \leftrightarrow l_1, c_4 \leftrightarrow c_3, l_4 \leftrightarrow l_3$. Then translate every point with the vector $(-1,-1)$. This transforms a set of points $X$ (left-hand side) into a lower set $L$ (right-hand side). $V(X,L)$ is invertible.
\end{example}

\begin{tikzpicture}[scale=0.4]
\draw[help lines, color=gray!30, dashed] (0,0) grid (5.9,5.9);
\draw[->,ultra thick] (0,0)--(6,0) node[right]{$x$};
\draw[->,ultra thick] (0,0)--(0,6) node[above]{$y$};
\foreach \Point in {(0,3), (1.5,3), (3,3), (4.5,3), (6,3), (1.5,1.5), (3,1.5), (4.5,1.5) ,(1.5,4.5), (3,4.5), (4.5,4.5), (3,0), (3,6)}{
    \node [orange] at \Point {\textbullet};
}
\node at (7.5,3) {$\longrightarrow$};
\end{tikzpicture} 
\begin{tikzpicture}[scale=0.4]
\draw[help lines, color=gray!30, dashed] (0,0) grid (5.9,5.9);
\draw[->,ultra thick] (0,0)--(6,0) node[right]{$x$};
\draw[->,ultra thick] (0,0)--(0,6) node[above]{$y$};
\foreach \Point in {(1,3), (2,3), (3,3), (4,3), (5,3), (2,2), (3,2), (4,2) ,(2,4), (3,4), (4,4), (3,1), (3,5)}{
    \node [orange] at \Point {\textbullet};
}
\node at (7,3) {$\longrightarrow$};
\end{tikzpicture} 
\begin{tikzpicture}[scale=0.4]
\draw[help lines, color=gray!30, dashed] (0,0) grid (5.9,5.9);
\draw[->,ultra thick] (0,0)--(6,0) node[right]{$x$};
\draw[->,ultra thick] (0,0)--(0,6) node[above]{$y$};
\foreach \Point in {(1,1), (1,2), (1,3), (1,4), (1,5), (2,1), (2,2), (2,3) ,(3,1), (3,2), (3,3), (4,1), (5,1)}{
    \node [orange] at \Point {\textbullet};
}
\node at (7,3) {$\longrightarrow$};
\end{tikzpicture}
\begin{tikzpicture}[scale=0.4]
\draw[help lines, color=gray!30, dashed] (0,0) grid (5.9,5.9);
\draw[->,ultra thick] (0,0)--(6,0) node[right]{$x$};
\draw[->,ultra thick] (0,0)--(0,6) node[above]{$y$};
\foreach \Point in {(0,0), (0,1), (0,2), (0,3), (0,4), (1,0), (1,1), (1,2) ,(2,0), (2,1), (2,2), (3,0), (4,0)}{
    \node [orange] at \Point {\textbullet};
}
\end{tikzpicture}

We can use this intuition to find arrangements of sensors that grant the Vandermonde matrix is invertible. In particular, let us describe a map that encompasses both of these transformations, for which we will establish an equivalence relation, allowing us to prove the invertibility.

\begin{definition}[Labeling Function]\label{def:labelingfunction}
Let $X$ define a set of points in a $m$ dimensional module over some commutative ring $\mathbb{K}$. Define a labeling function as a map:
\begin{equation}
\begin{aligned}
	a : \mathbb{Z}^m &\longrightarrow X \subseteq \mathbb{K}^m \\
	z = (z_1,\cdots,z_m) &\longmapsto a_z \equiv a(z) = x \in X
\end{aligned}
\end{equation}
Holding the following properties:
\begin{enumerate}
	\item $a(z_1,\cdots,z_m) = (a_1(z_1),a_2(z_2),\cdots,a_m(z_m)) $ where all $a_j$ are independent of each other
	\item $a^{-1} (x)$ is uniquely defined, i.e. corresponds to a bijection, for every $x$
	\item Let $a,b$ be labeling functions. Then $a \circ b$  and $b \circ a$ are labeling functions
\end{enumerate}
Note property 3 is a direct consequence of 1 and 2.
\end{definition}

Let us establish a relation between a set of points in a arbitrary space $\mathbb{K}^m$ and a set of points over the integer module $\mathbb{Z}^m$. Let $X \subseteq \mathbb{K}^m, Z \subseteq \mathbb{Z}^m $, we say $X \sim Z$ if one can find a labeling function $a$ such that $a(Z)=X$. Under this relation we can find a more general theorem for the invertibility of our Vandermonde matrix guaranteeing our system will have a solution:

\begin{theorem}
Let $X \subseteq \mathbb{R}^m$, with $|X|=p$, such that $X \sim Y$ for some lower set $Y$ of $(\mathbb{N}_0^n,\preceq)$. Then the Vandermonde matrix $V(X,Y) \in \mathbb{R}^{p \times p}$ is invertible.
\end{theorem}

\begin{proof}
The proof follows from Example~\ref{example:borderideals}. Given that $X \sim Y$, then there exists a bijection $a$ mapping $Y$ to $X$. This means that one is able to build a set of ideals as in example~\ref{example:borderideals}, substituting $x - a_k$, $y - b_k$, ... by corresponding $x_j - a_j(y)$ for some $ y \in Y$ , as $Y$ is a lower set and therefore has a Pareto front with size at least 1. From here we recover invertibility as the degrees of the ideals are given by the cover of $Y$, which means we are able to choose a set of monomials $\prod_{j=1}^n x_j^{\alpha_j}$ for all $\alpha \in Y$ such that the points in $X$ never belong to a variety spanned by these monomials. From Theorem~\ref{thm:invertibility1}, we then get invertibility of the Vandermonde matrix $V(X,Y)$.
\end{proof}

Note however that the relabeling of a set is not unique, and this has consequences. Namely, one can relabel a set $X$ into two different lower sets $L_1 \neq L_2$ that constructs two different Vandermonde matrices. For example:

\begin{example}
Let $X \subset \mathbb{N}_0^2$, be given by $X = \{ (j,j) | j=1,2,\cdots, p \}$. There are two possible relabelings, that give access to all derivatives up to $p-1$ in either one or the other direction, but never any cross derivatives:
\begin{center}
\begin{tikzpicture}[scale=0.4]
\draw[help lines, color=gray!30, dashed] (0,0) grid (5.9,5.9);
\draw[->,ultra thick] (0,0)--(6,0) node[right]{$x$};
\draw[->,ultra thick] (0,0)--(0,6) node[above]{$y$};
\foreach \x in {1,2,3,5}{
	\node[orange] at (\x,\x) [circle,fill,inner sep=1.5pt]{};
}
\draw[-,thick,orange] (4,4)--(4,4) node[]{$\rddots$};
\draw[-,thick] (7,3)--(7,3) node[above]{$ \sim $};
\end{tikzpicture} 
\begin{tikzpicture}[scale=0.4]
\draw[help lines, color=gray!30, dashed] (0,0) grid (5.9,5.9);
\draw[->,ultra thick] (0,0)--(6,0) node[right]{$x$};
\draw[->,ultra thick] (0,0)--(0,6) node[above]{$y$};
\foreach \x in {1,2,3,5}{
	\node[orange] at (\x,0) [circle,fill,inner sep=1.5pt]{};
}
\draw[-,thick,orange] (4,0)--(4,0) node[]{$\cdots$};
\draw[-,thick] (7,3)--(7,3) node[above]{$ \sim $};
\end{tikzpicture} 
\begin{tikzpicture}[scale=0.4]
\draw[help lines, color=gray!30, dashed] (0,0) grid (5.9,5.9);
\draw[->,ultra thick] (0,0)--(6,0) node[right]{$x$};
\draw[->,ultra thick] (0,0)--(0,6) node[above]{$y$};
\foreach \x in {1,2,3,5}{
	\node[orange] at (0,\x) [circle,fill,inner sep=1.5pt]{};
}
\draw[-,thick,orange] (0,4)--(0,4) node[]{$\vdots$};
\end{tikzpicture} 
\end{center}
\end{example}

To conclude this section, we have found how the choice of set $X$ (position of the sensors) impacts the invertibility of the Vandermonde matrix, given an expansion using a lower set $L$. This means how we can find solutions for estimators of the expansion coefficients, which contain no error from the non-invertibility of the Vandermonde matrix. We found a sufficient condition: when the expansion orders forming a lower set $L$ and the set of points where the sensors are located $X$ are equivalent under relabeling. It is not a necessary condition, as one could think of relaxing some assumptions on the relabeling function, that would still imply invertibility of the Vandermonde matrix. We leave this for future work.

\subsection{Interpolation Problem and Linear Combinations}

So far we have found the way to transform the problem of evaluating a derivative of a multivariate function in a subset of points, namely those of the location of the sensors. This was equivalent to what is known in numerical differentiation of using finite differences methods, and consists of finding the solution to a specific linear system of equations (Eqs.~\ref{eq:vandermonde1D} and \ref{eq:vandermonde}). 

To find an interpolation of the function and evaluate it outside of the set of points, $i.e.$ finding the interpolation value of $F(\boldsymbol{x}_t)$, given by the sensor locations can be made two-fold: either via finding the Taylor expansion around a sensor point $\boldsymbol{x}_{\min}$ close to the target $\boldsymbol{x}_t$, or by finding the Taylor expansion around the target point $\boldsymbol{x}_t$. From the previous methods, it is more intuitive to understand why the first approach would work, because the result of Eqs.~\ref{eq:expansions} and Eqs.~\ref{eq:expansionsMulti} are the derivatives at a point where we have one sensor, which correspond to the coefficients of the Taylor expansion around that same point. Let us start by describing the problem and a method to find the solutions in this case:

\begin{problem}\label{problem:interpolation}
Let $X \subset D \subseteq \mathbb{R}^m$ be a set of points or sensor locations in $m$ dimensions. Let $\boldsymbol{x}_t \not\in X$ be a point in the domain of a field $F(\boldsymbol{x})$, $\boldsymbol{x}_t \in D$. Find the interpolation value of $F(\boldsymbol{x}_t)$ given access to $\boldsymbol{F}= \{ F(\boldsymbol{x})\}_{\boldsymbol{x} \in X}$.
\end{problem}

Problem~\ref{problem:interpolation} can be solved in the following way:
\begin{enumerate}
	\item Find the closest point $\boldsymbol{x}_{\min} = \min_{\boldsymbol{x} \in X}\norm{\boldsymbol{x}_t - \boldsymbol{x}}_\infty$ and construct $\delta X = \{\boldsymbol{x} - \boldsymbol{x}_{\min} \}_{\boldsymbol{x} \in X}$;
	\item Find a lower set $L$ in $(\mathbb{N}_0^m)$ that is equivalent under relabeling to $X$, $L \sim X$
	\item Expand the field around $\boldsymbol{x}_{\min} $:
	\begin{equation}
		F(\boldsymbol{x}_t ) \approx  \sum_{\boldsymbol{\alpha} \in L} \frac{ D^{\boldsymbol{\alpha}} F(\boldsymbol{x}_{\min}) }{\boldsymbol{\alpha} !}(\boldsymbol{x}_t-\boldsymbol{x}_{\min})^{\boldsymbol{\alpha}} 
	\end{equation}
	\item Using the field values for the points in $X$, use Eq.~\ref{eq:vandermonde} to find an $L$ approximation of $ D^{\boldsymbol{\alpha}} F(\boldsymbol{x}_{\min})$:
	\begin{equation}
	\begin{aligned}
		F(\boldsymbol{x}_t ) &\approx  \sum_{\boldsymbol{\alpha} \in L} \frac{ \left( V(\delta X,L)^T \right)^{-1} \boldsymbol{\alpha}! \boldsymbol{e}_{\boldsymbol{\alpha}} \cdot \boldsymbol{F} }{\boldsymbol{\alpha} !}  (\boldsymbol{x}_t-\boldsymbol{x}_{\min})^{\boldsymbol{\alpha}}  \\
		&=  \left( V(\delta X,L)^T \right)^{-1} \left(  \sum_{\boldsymbol{\alpha} \in L} \boldsymbol{e}_{\boldsymbol{\alpha}} (\boldsymbol{x}_t-\boldsymbol{x}_{\min})^{\boldsymbol{\alpha}} \right) \cdot \boldsymbol{F} \\
		&=\left( \left( V(\delta X,L)^T \right)^{-1} \boldsymbol{r} \right) \cdot \boldsymbol{F} = \boldsymbol{c} \cdot \boldsymbol{F}
	\end{aligned}
	\end{equation}
	where we used $\boldsymbol{r}$ to be a vector with the same size as $L$, such that the component $\boldsymbol{\alpha}$,  $(\boldsymbol{r})_{\boldsymbol{\alpha}} = (\boldsymbol{x}_t-\boldsymbol{x}_{\min})^{\boldsymbol{\alpha}}$. Notice all the sizes match, $X$ is a set with $|X|=p$ elements, which means $L$ and $\boldsymbol{F}$ have $p$ elements, which in turn mean $V \in \mathbb{R}^{p\times p}$ and $\boldsymbol{r} \in \mathbb{R}^p$.
\end{enumerate}

Next let us generalize this for any linear combination of the interpolation of a derivative. This is simply applying the same as before, but to the expansion presented in Eq.~\ref{eq:expansionderivatives}:

\begin{problem}\label{problem:interpolationderivatives}
Let $X \subset D \subseteq \mathbb{R}^m$ be a set of points or sensor locations in $m$ dimensions. Let $\boldsymbol{x}_t \not\in X$ be a point in the domain of a field $F(\boldsymbol{x})$, $\boldsymbol{x}_t \in D$. Find the interpolation value of a derivative $D^{\boldsymbol{\zeta}}F(\boldsymbol{x}_t)$ given access to $\boldsymbol{F}= \{ F(\boldsymbol{x})\}_{\boldsymbol{x} \in X}$.
\end{problem}

Problem~\ref{problem:interpolationderivatives} can be solved very similarly by redefining vector $\boldsymbol{r}$ in the following way:

\begin{equation}
	( \boldsymbol{r} )_\alpha = \begin{cases}
	(\boldsymbol{x}_t-\boldsymbol{x}_{\min})^{\boldsymbol{\alpha}-\boldsymbol{\zeta}} \frac{ \boldsymbol{\alpha}!}{(\boldsymbol{\alpha} - \boldsymbol{\zeta})!} , \qquad & \zeta \preceq \alpha \\
	0 , \qquad & \text{ otherwise }
	\end{cases}
\end{equation}

In both cases we are able to transform the problem in a linear combination of the values of the field at a set of points. Naturally, a linear combination of linear combinations is still a linear combination. This means that any function that can be written as a linear combination of a set of interpolation points, or interpolation of the derivatives, is also accessible as a linear combination of the problems above.

As we mentioned in the beginning of this section, there exists an alternative way of doing this: expanding directly over the target point. Let us detail how to solve Problem~\ref{problem:interpolationderivatives} in this case (note that the case of $\boldsymbol{\zeta} = \vec{0}$ corresponds to the interpolation value, and is in fact Problem ~\ref{problem:interpolation}):
\begin{enumerate}
	\item Construct $\delta X = \{\boldsymbol{x} - \boldsymbol{x}_t \}_{\boldsymbol{x} \in X}$;
	\item Find a lower set $L$ in $(\mathbb{N}_0^m)$ that is equivalent under relabeling to $X$, $L \sim X$
	\item Using the field values for the points in $X$, use Eq.~\ref{eq:vandermonde} directly:
	\begin{equation}
	\begin{aligned}
		D^{\boldsymbol{\zeta}}F(\boldsymbol{x}_t ) &\approx \left[  \left( V(\delta X,L)^T \right)^{-1} \boldsymbol{\zeta}! \boldsymbol{e}_{\boldsymbol{\zeta}} \right] \cdot \boldsymbol{F}
	\end{aligned}
	\end{equation}
\end{enumerate}

A few examples of approximations that one can by employing these techniques are the laplacian of a function on a point, the determinant of the Hessian matrix at a given point, the average of the function at a given set of points and so on.

\subsubsection{Error of Derivatives/Interpolation}

In all of the above calculations and theorems, from the analyticity of the field, we assumed that the Taylor expansion converged. However, if one only takes a finite amount of points, then one can only find an approximation that is as good as the amount of points taken. For this reason, let us give a bound for the error is the interpolation or target function. As we seen above, in all problems we are able to write the target function $D^{\boldsymbol{z}} F(\boldsymbol{y})$ as $\boldsymbol{c}(\boldsymbol{z},\boldsymbol{y})\cdot \boldsymbol{F}$. Plugging this $\boldsymbol{c}(\boldsymbol{z},\boldsymbol{y})$ back into the original expansion around the set of points $X$:

\begin{equation}
\begin{aligned}
	D^{\boldsymbol{z}} F(\boldsymbol{y}) = \boldsymbol{c}(\boldsymbol{z},\boldsymbol{y})\cdot \boldsymbol{F} + \boldsymbol{c}(\boldsymbol{z},\boldsymbol{y})\cdot \boldsymbol{\varepsilon}(\boldsymbol{z},\boldsymbol{y})
\end{aligned}
\label{eq:error}
\end{equation}
Where $\boldsymbol{\varepsilon}(\boldsymbol{z},\boldsymbol{y})$ is the vector where each component correspondent to each point is given by $(\boldsymbol{\varepsilon}(\boldsymbol{z},\boldsymbol{y}))_j = \epsilon (\boldsymbol{z},\boldsymbol{x}_j, \boldsymbol{y})$ of Eq.~\ref{eq:expansions}. This corresponds to the individual expansion error of the field around the point $\boldsymbol{x}_j \in X$. Note that the vector $\boldsymbol{c}$ one would obtain from inverting the Vandermonde matrix should intuitively counteract when the corresponding $\boldsymbol{x}_j $ is distant from the expansion point. On the other side, if one tries to interpolate outside of the neighborhood of $X$, $i.e.$ a point $\boldsymbol{y}$ such that $\norm{\boldsymbol{y}-\boldsymbol{x}} \gg 1, \forall \ x\in X$, then the interpolation has a larger error, as expected. 

Additionally, since this is equivalent to an interpolation, if the field is given by a polynomial function whose monomial degrees only take values over a subset which is equivalent under relabeling to $X$, then the interpolation is exact and there is no error of approximation. Taking $L\sim X$ as the set for the degrees of the monomials, this can be seen as the derivatives with degree on $\mathsf{Cov}(L)$ being zero and the expansion being exact.

\subsection{Example: Non-Local \textit{vs.} Local Strategy on a 2D Grid Interpolation}\label{section:interpolationexample}

We now explore the possible advantage in an entangled, or what we call 'non-local' strategy and a local strategy that does not make use of entanglement.

As an example for the interpolation let us consider the following 2D grid interpolation problem:

\begin{problem}
Suppose we have an arrangement of sensors of $n\times n$ sensors, arranged as a grid. Moreover, assume that our field is given by a polynomial function:
\[
	F(\boldsymbol{x}) = (x -1)^k + (y-1)^k
\]
Given the sensor positions $X = \{\boldsymbol{x}_1, \boldsymbol{x}_2, \dots , \boldsymbol{x}_{n^2}\}$, and local evaluation of the field $\boldsymbol{F} = \{F(\boldsymbol{x})\}_{\boldsymbol{x} \in X}$, find the value of $F(\boldsymbol{x}_t)$ for a given target location $\boldsymbol{x}_t$.
\end{problem}

Using the methodology introduced in this section, we verify that the Taylor expansion available will have errors for $\mathcal{O}(\boldsymbol{x}^{\boldsymbol{\alpha}})$ where $\norm{\boldsymbol{\alpha}}_\infty \geq n$ (see Example \ref{ex:lowerset}, leftmost lower set). This means that the model would be error free if $n>k$. To find the value of $F(\boldsymbol{x}_t)$ one would first build the Vandermonde matrix at each point $\boldsymbol{x}_t$ and then build the corresponding vector such that:

\begin{equation}
	F(\boldsymbol{x}_t) = \left[(V(\delta X_t)^T)^{-1} \boldsymbol{e}_0 \right]\cdot \boldsymbol{F} = \boldsymbol{c} \cdot \boldsymbol{F}, \ \ \delta X_t = \{\boldsymbol{x}-\boldsymbol{x}_t \}_{\boldsymbol{x} \in X},  \boldsymbol{c} =  \boldsymbol{c}(\boldsymbol{x}_t)
\end{equation}

If one uses $3\times 3$ sensors, with $k=3$ one would have errors amounting to $\mathcal{O}(\boldsymbol{x}^{\boldsymbol{\alpha}})$ where $\norm{\boldsymbol{\alpha}}_\infty = 3$, but by choosing an array of $5 \times 5$ sensors, than the estimation would be error free, as the polynomial has maximum degree $\norm{\boldsymbol{\alpha}}_\infty = 5$. 

\begin{figure}[t]
\centering
\subfloat{\includegraphics[width=0.75\linewidth]{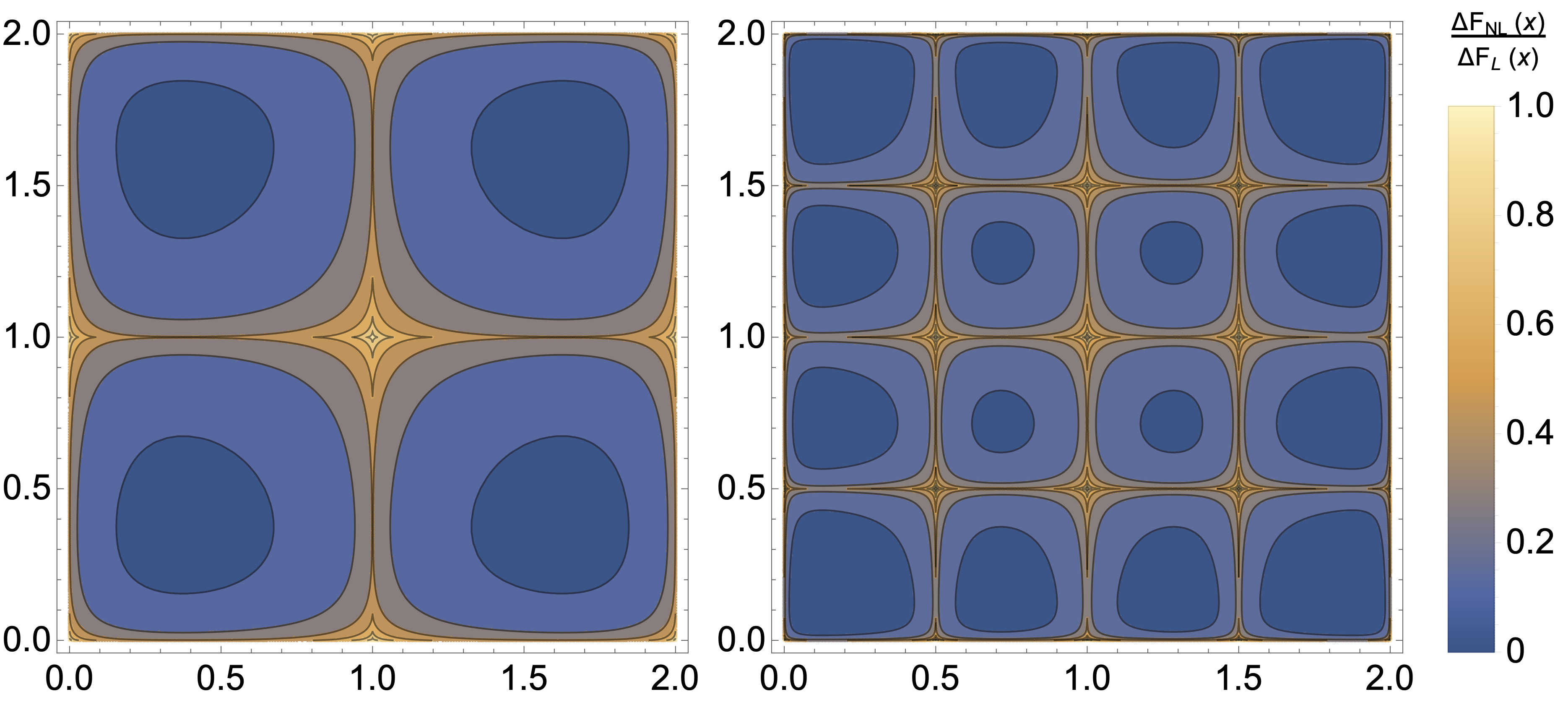}} 
\caption{Comparison between the best non-local and local strategies for the solution of a interpolation problem with $k=3$, using a equally spaced grid of sensors in two dimensions made of $3\times 3$ sensors (left) and $5\times 5$ sensors (right). To each point $\boldsymbol{x}_t = (x,y)$ we calculate and plot the gain on the precision (measured by the variance of the non-local divided by the variance of the local) of employing a non-local strategy to estimate $\hat{F}(\boldsymbol{x}_t)$. }
\label{fig:variancecomparison}
\end{figure}

Linking with a quantum estimation, we know that a private state, and in particular a GHZ state whose resource distribution we are able to control, is the state that allows to retrieve the most information \cite{Bugalho2025}. In particular, comparing among all non-local strategies fixing only the amount of total resources, we get the GHZ state distributed according to the vector $\boldsymbol{c}$ is the optimal non-local solution. 

One can also figure out which is the best local strategy, where by local we mean that only qubits inside of the same sensor can be entangled. This means there cannot exist any entanglement between the various sensors. For this, one can also use GHZ states locally, which are known to achieve the Heisenberg scaling, and get the most information about a target parameter, which in this case will be the local values of the field $F(\boldsymbol{x}_j)$. Then we can use these local estimators to build $\boldsymbol{c} \cdot \boldsymbol{F}$. While in the non-local case, the optimal distribution of resources is given directly by the vector $\boldsymbol{c}$, for the local strategy this will not be the case. In Appendix \ref{appendix:optimizing} we prove that the resource distribution, represented by a vector of resources $\boldsymbol{n}$, that minimizes the variance of the estimator of $F(\boldsymbol{x}_t) = \boldsymbol{c}\cdot \boldsymbol{F}$, $\Delta \hat{F}(\boldsymbol{x}_t)^2$, for fixed amount of total resources, $N$, is given by:
\begin{equation}
\begin{aligned}
	\boldsymbol{n}_{\NL} &= \frac{N}{\norm{\boldsymbol{c}}_1} \boldsymbol{c} \\
	\Delta \hat{F}(\boldsymbol{x}_t)^2_{\NL} &= \frac{ \norm{\boldsymbol{c}}_1^2 }{m N^2}
\end{aligned}
 \qquad , \qquad
\begin{aligned}
	\boldsymbol{n}_{\LL} &= \frac{N}{\norm{\boldsymbol{c}}_{2/3}^{3/2}} \boldsymbol{c}^{2/3} \\
	\Delta \hat{F}(\boldsymbol{x}_t)^2_{\LL} &= \frac{ \norm{\boldsymbol{c}}_{2/3}^2}{m N^2}
\end{aligned},
\label{eq:optimal}
\end{equation}
where we used that $\boldsymbol{c}^d$ is the entry-wise power of a vector, $i.e.$ where each $(\boldsymbol{c}^d)_j = (\boldsymbol{c})_j^d$. Note that from the way the estimators are constructed, the optimal non-local estimation has access to information only about the target function, and therefore is private. In the local case, each sensor would first find the best estimator for its value of the field, and then they would share it with the rest of the parties, so they could build the estimator for $F(\boldsymbol{x}_t)  = \boldsymbol{c}\cdot\boldsymbol{F}$.

To exemplify in one scenario addressed in this paper, we plotted in Fig.~\ref{fig:variancecomparison} the optimal non-local vs. the optimal local estimator of the interpolation value of $\hat{F}(\boldsymbol{x}_t)$ for $\boldsymbol{x}_t$ inside the domain of a 2-dimensional regular grid with $3\times 3 = 9$ sensors and with $5\times 5=25$ sensors. We observe that the entangled strategy is always better  than the non-entangled strategy, except for estimating $F(\boldsymbol{x}_j)$ where $\boldsymbol{x}_j$ belongs to one of the points where we have a sensor. This brings back the intuition that if we want to measure a quantity that is local, then the best strategy is local. This can also be seen as a consequence of $\norm{\vec{c}}_1 \leq \norm{\vec{c}}_{2/3} \leq \sqrt{d} \norm{\vec{c}}_1$, where $d$ is the length of the vector, therefore $n^2$ in this case. The left equality happens at $\vec{c} = \alpha \vec{e}_j$ (local) and the right equality happens if $\vec{c} = \alpha \vec{1}$ (average). Note that from Eq.~\ref{eq:optimal}, this only depends on the positions and number of sensors, not the actual function itself. Moreover, since this is a polynomial interpolation, if $k<n$, then the interpolation is exact. In Appendix~\ref{appendix:interpolation} we provide plots of the interpolation error for different polynomials, and discuss more on the accuracy of the interpolation.

\section{Analytical Functions}\label{section:signal}

We started by looking at Taylor expansions of functions, and using the Taylor expansions and polynomials to interpolate the field at arbitrary points. This was equivalent as limiting our subset of functions $\mathcal{F}$ in Eq.~\ref{eq:fielddefinition} to the space of polynomial functions, and in particular to monomials of increasing degree. This came as a consequence of having no \textit{a-priori} information about the field in question, which implied that there should be no privileged set of functions that could describe our field. This in turn gave us information about the numerical derivatives of the field, and a method to build an interpolation of the field at any point, together with the derivatives of the interpolation. 

As one gains more information about the field, $e.g.$ the constituent signals that describe the field, then one can think how to do better than before. This is the case where the functions in the model of our field can be chosen from available information, and in particular one can use a set of functions that is the same size as the set of points, or data, available. 

Suppose we have again a set of points with sensors $X$, such that at each point $x \in X$ we have access to $F(x)$. In this case the functions in the model we want to test, that can be thought exactly as signals following the language of \cite{Sekatski2019,Hamann2022,Hamann2023}. The problem can then be formulated into a signal-isolation of a family of functions, or signals, $\mathcal{F} = \{f_1, f_2, \cdots , f_k\}$, which are continuously differentiable on some domain $D, X\subseteq D$, describing $F$. These signals will create the field $F(x)$ by a linear combination of these functions: 

\begin{equation}
	F(x) = \sum_{j=1}^k \beta_j f_j (x).
	\label{eq:field}
\end{equation}

Given a set of $k$ points $X$, one can define the alternant matrix $A=A(X,\mathcal{F})$ as the matrix described by the following equation:

\begin{equation}
	(A)_{ij} = f_i (x_j), \qquad f_i \in \mathcal{F}, x_j \in X.
\end{equation}

One can see that by taking the monomials as the set for $\mathcal{F}$, one can create the Vandermonde matrix. In this case our goal is to find the vector $\beta$, which corresponds to the different weights of each of the signals that best describes the observed data. This can be thought as isolating the coefficients for each of the signals present. The solution should be given by solving the following system:

\begin{equation}
	\boldsymbol{\beta} = A^{-1} \boldsymbol{F}, \qquad \text{where } (\boldsymbol{F} )_j = \hat{F}(x_j).
\label{eq:solAlternant}
\end{equation}

Once again, we find a condition for the invertibility of the matrix $A$ in terms of the sensors positions and the family of functions, in the same way we found for the Vandermonde matrix. However, this time we do not provide a methodology to build a set of points, or positions for the sensors, that result in an invertible matrix. The main reason behind this is that the problem is highly dependent on the functions in the model. We discuss further about the Alternant matrix and its invertibility in Appendix \ref{appendix:alternant}. Nonetheless, and in the same way as before, an invertible alternant matrix results from a set of positions that are able to discern unambiguously the model in question.

\subsection{Signal-Isolation Problem}

As before the interpolation problem had a direct output of identifying the derivatives of the field at some point, the signal-isolation problem also has a direct outcome -- calculating the relative strength of each of the signal sources, of model functions. This are the coefficients $\beta_j$ that appear in Eq.~\ref{eq:field}. In fact, if one substitutes the model functions by the monomials, then the relative strengths of these "signals" are exactly the derivatives at the expansion point. Let us then state this problem, which is the same problem analyzed in \cite{Sekatski2019,Hamann2022,Hamann2023,Hamann2024}:

\begin{problem}\label{problem:signalisolation}
Let $X \subset D \subseteq \mathbb{R}^m$ be a discrete set of points or sensor locations in $m$ dimensions of size $p$. Let $\boldsymbol{\beta}  \in \mathbb{R}^p, \mathcal{F} = \{ f_1,f_2,\dots,f_p\}$ be the coefficients and the functions corresponding to the field signals (see Eq.~\ref{eq:field}). Find the estimator of the coefficient $\beta_t$ given access to $\boldsymbol{F}= \{ F(\boldsymbol{x})\}_{\boldsymbol{x} \in X}$.
\end{problem}

Problem~\ref{problem:signalisolation} can be solved in the following way:
\begin{enumerate}
	\item Construct the alternant matrix $A = A (X,\mathcal{F})$ such that $(A)_{ij} = f_j(x_i)$
	\item Using the inverse of the alternant matrix:
	\begin{equation}
	\begin{aligned}
		\beta_t &= \left[ (\boldsymbol{A}^T)^{-1} \boldsymbol{e}_t \right] \cdot \boldsymbol{F} \\
		&= \boldsymbol{c} \cdot \boldsymbol{F}
	\end{aligned}
	\end{equation}	
\end{enumerate}

From here, we can in the same way as before, generalize to find any value of the field at any position of the domain, including any derivatives. This is done by redefining a vector of weights, by noticing that the field is described by a sum of the model functions weighted by the coefficients, which one can extract with the above problem. Therefore:

\begin{problem}\label{problem:fieldeval}
Let $X \subset D \subseteq \mathbb{R}^m$ be a discrete set of points or sensor locations in $m$ dimensions of size $p$. Let $\boldsymbol{\beta}  \in \mathbb{R}^p, \mathcal{F} = \{ f_1,f_2,\dots,f_p\}$ be the coefficients and the functions corresponding to the field signals (see Eq.~\ref{eq:field}). Find the model value of a derivative $D^{\boldsymbol{\zeta}}F(x_t)$ given access to $\boldsymbol{F}= \{ F(\boldsymbol{x})\}_{\boldsymbol{x} \in X}$.
\end{problem}

Problem~\ref{problem:signalisolation} can be solved in the following way:
\begin{enumerate}
	\item Construct the alternant matrix $A = A (X,\mathcal{F})$ such that $(A)_{ij} = f_j(x_i)$
	\item Build a weight vector given by $\boldsymbol{b} = \{ D^{\boldsymbol{\zeta}}f_j (x_t), \text{for } f_j \in \mathcal{F}\}$
	\item Using the inverse of the alternant matrix:
	\begin{equation}
	\begin{aligned}
		D^{\boldsymbol{\zeta}}F(x_t) &= \sum_j \beta_j b_j \\
		&= \sum_j \left[ (\boldsymbol{A}^T)^{-1} \boldsymbol{e}_j b_j \right] \cdot \boldsymbol{F} \\
		&= \left[ (\boldsymbol{A}^T)^{-1} \boldsymbol{b} \right] \cdot \boldsymbol{F} \\
		&= \boldsymbol{c} \cdot \boldsymbol{F}
	\end{aligned}
	\end{equation}	
\end{enumerate}

As before, one can do any linear combination of the above problems, and still find a way to describe the problem as a linear combination of the field evaluated at different points. It is important to notice that in this case, the alternant matrix is a square matrix. This means we are putting as many functions in our model, as we can. This somehow relates to having some information about the functions, but not necessarily knowing to the fullest extent which are the ones that better describe our field. If one knew that our field can only contain a subset of $k < p$ functions, then more information is available, and this is exactly the subject of the next section.

\subsection{Example: Magnetic Field Source Finding}

As an example of a signal-isolation problem let us consider the following magnetic field source finding problem:

\begin{problem}
Suppose we have 4 sensors (blue dots), at the corners of a square of side $L$, $X = \{ x_1, x_2, x_3, x_4 \}$. There exists a magnetic field source passing inside this square at an unknown position, from the set of 4 possibilities $Z =  \{z_1, z_2, z_3, z_4\}$ (red dots). This corresponds to a model of the field given by:
\[
	\hat{F}(\boldsymbol{x}) = \sum_{i=1}^4 \beta_i \cdot \frac{1}{\norm{\boldsymbol{x}-\boldsymbol{z}_i}}
\]
\begin{center}
\begin{tikzpicture}
    \draw[very thin, gray,dashed] (-0.25,-0.25) grid (2.25,2.25);
    \foreach \x in {0, 2 } {
    	\foreach \y in {0,2} {
       		 \fill[blue] (\x,\y) circle (2pt);
	}
    }
     \foreach \x in {0.5, 1.5 } {
    	\foreach \y in {0.5,1.5} {
       		 \fill[red] (\x,\y) circle (2pt);
	}
    }
   
\end{tikzpicture}
\end{center}

Given the sensor positions, and local evaluation of the field $\boldsymbol{F} = \{F(x)\}_{x \in X}$, find if the origin source is at $z_2$.
\end{problem}

Using the methodology introduced in this section, the problem solution comes from finding the estimator for the value $\beta_2$. This is simply given from building the alternant matrix, and then the vector $\boldsymbol{c}$ in the following way:
\begin{equation}
	\beta_2 = \left[ (\boldsymbol{A}^T)^{-1} \boldsymbol{e}_2 \right] \cdot \boldsymbol{F} = \boldsymbol{c} \cdot \boldsymbol{F}, \ \  (A)_{ij} =  \frac{1}{\norm{\boldsymbol{x}_i-\boldsymbol{z}_j}}
\end{equation}

\section{Generalized Least-Squares}\label{section:leastsquares}

Curiously, the applications of this framework do not end here. Looking at something from the statistics field, one realizes that there is a modeling tool which is widely used for statistical regression, that is also linear in terms of the values of the local field values. Let us write the problem in question:

\begin{problem}[Generalized Least-Squares (GLS)]\label{prob:GLS}
Let $\{ y_i, x_{ij}\}$ be a set of observed data and predictor values, represented by a vector of size $p$, $\boldsymbol{y}$ and a matrix of size $p\times k$, $\boldsymbol{X}$. Let:
\begin{equation}
	\mathbf{y}=\mathbf{X} \boldsymbol{\beta}+\boldsymbol{\varepsilon}, \quad \mathbb{E}[\varepsilon \mid \mathbf{X}]=0, \quad \operatorname{Cov}[\varepsilon \mid \mathbf{X}]=\boldsymbol{\Omega}
\end{equation}
The outcome of the problem is a vector $\hat{\boldsymbol{\beta}}$ such that $\mathbb{E} [ \hat{\boldsymbol{\beta}} | \boldsymbol{X}] = \boldsymbol{\beta}$ satisfying the following equation:
\begin{equation}
\hat{\boldsymbol{\beta}}=\underset{\mathbf{b}}{\operatorname{argmin}}(\mathbf{y}-\mathbf{X} \mathbf{b})^{\mathrm{T}} \boldsymbol{\Omega}^{-1}(\mathbf{y}-\mathbf{X} \mathbf{b})
\end{equation}
\end{problem}

The vector $\boldsymbol{y}$ is given by all the $y_i$, which are the measured values, or the observed data. In our estimation scenario, these are the local values of the field $F(x_j)$, which we have assembled in the vector $\boldsymbol{F}$. The matrix $\boldsymbol{X}$ is called the design matrix, and it is an important concept in statistics. This dates all the way back to the design of experiments to test hypothesis, a field pioneered by Ronald Fisher, the name behind the Fisher information. In particular, take $(\boldsymbol{X})_{ij}$ as the $j$th feature evaluated at point $i$: the feature is a function $\phi_j$ evaluated over a point $x_i$, $\phi_j(x_i)$. Taking features from the space of analytical functions defined over some domain of points $\phi_j \rightarrow f_j$ and $x_i \in X$ and taking the size of samples $n$ to be the same as the size of features $p$ retrieves the case for the alternant matrix analyzed before. Nevertheless, the features are not limited to analytical functions, but can be used to characterize features of different populations, which was the basis of analysis of variance (ANOVA) tests introduced by Fisher.

We have outlined the generalized least-squares method. One can recover the regular linear least-squares method by making the weight correlation matrix $\boldsymbol{\Omega} \rightarrow \mathbb{1}$. If instead we allow it to be any diagonal matrix $\boldsymbol{W}$, we recover the weighted least-squares method. The important difference between this and the previously outline method for analytical functions is the functions model. In this case we go beyond not knowing the exact model, but searching over the linear span of a set of functions that represent our model, that is smaller than what our points are able to tell. For this to be valid, then one needs to know in fact what functions better represent the model for the underlying field. This is the same as having more information, so we can better refine what we are looking. At the same time, we can test among the data more than the underlying field, but also statistical properties of the data, extending therefore the notion of what is in fact a global field, to being able to discern arbitrary data features.

Importantly, the solution of Problem~\ref{prob:GLS} is given by the vector that better reproduces the observed values, meaning that it is the closest to the predictor values, by minimizing the vector $\varepsilon$. The solution for Problem \ref{prob:GLS} is given by:

\begin{equation}
\hat{\boldsymbol{\beta}}=\left(\mathbf{X}^{\mathrm{T}} \boldsymbol{\Omega}^{-1} \mathbf{X}\right)^{-1} \mathbf{X}^{\mathrm{T}} \boldsymbol{\Omega}^{-1} \mathbf{y}
\label{eq:solGLS}
\end{equation}

Since $\boldsymbol{y}$ can be the vector of the field evaluated at different positions, $\{ F(\boldsymbol{x}_j ) \}_{j=1,\cdots , n} \equiv \boldsymbol{F}$, then we are able to write the solution of each of the coefficients of a GLS problem as a linear combination of the vector $\boldsymbol{F}$ which is our goal. Note that $\left(\mathbf{X}^{\mathrm{T}} \mathbf{X}\right)^{-1} \mathbf{X}^{\mathrm{T}} \equiv \boldsymbol{X}^+ $ is what is known as the pseudo-inverse of $\boldsymbol{X}$ (in the case $\boldsymbol{\Omega} = \mathbb{1}$). Note that in this case, Eqs.~\ref{eq:solGLS} and ~\ref{eq:solAlternant} are equivalent, where the inverse was substituted by a pseudo-inverse, since the matrix is no longer necessarily a square matrix, and the vector of field values $ \boldsymbol{F}$ was substituted by the vector of observed data $ \boldsymbol{y}$. In turn, this means that $\boldsymbol{X}^+ \boldsymbol{X} = \mathbb{1}$ iff. its columns are linearly independent, which is the same type of condition we had in the alternant matrix. This is the condition for having again a model that is discernible from the collected data, or the set of field evaluated points.

\subsection{Least-Squares Problem}

Let us define exactly then this problem, which comprises both the interpolation problem (by making the model functions monomials with crescent degree) and signal-isolation problem (by choosing the set of model functions with some \textit{a-priori} information and size of the number of points, or sensors, available).

\begin{problem}\label{problem:leastsquares}
Let $X \subset D \subseteq \mathbb{R}^m$ be a discrete set of points, data or sensor positions in $m$ dimensions of size $p$. Let $\boldsymbol{b} \in \mathbb{R}^k, p \geq k$ be a vector with the same length has our vector of unknown $\boldsymbol{\beta}$ that are the coefficients of the field signals (see Eq.~\ref{eq:field}). Find the general least-squares estimator of $\boldsymbol{b}\cdot \boldsymbol{\beta}$ given access to $\boldsymbol{F}= \{ F(\boldsymbol{x})\}_{\boldsymbol{x} \in X}$.
\end{problem}

Problem~\ref{problem:leastsquares} can be solved in the following way:
\begin{enumerate}
	\item Construct the design matrix $\boldsymbol{X} = \boldsymbol{X} (X,\mathcal{F})$ such that $(\boldsymbol{X})_{ij} = f_j(x_i)$
	\item Find the generalized pseudo-inverse of $\boldsymbol{X}^+$, given a weight correlation matrix $\boldsymbol{\Omega}$
	\item Using the pseudo-inverse do:
	\begin{equation}
	\begin{aligned}
		\boldsymbol{b}\cdot \hat{\boldsymbol{\beta}} &= \boldsymbol{b}^T \boldsymbol{X}^+ \boldsymbol{F} \\
		&= \boldsymbol{c} \cdot \boldsymbol{F}
	\end{aligned}
	\end{equation}	
\end{enumerate}

In the problem above, in the case the weight matrix is the identity and the matrix is not full rank, while a generalized pseudo-inverse still exists, it will be non-unique. In particular, there will exist multiple pseudo-inverses such that $\boldsymbol{X}^+ \boldsymbol{X}$ is a diagonal matrix taking only ones and zeros in the diagonal. The number of ones will naturally depend on the rank of the matrix $\boldsymbol{X}$. Moreover, the rank should not be affected as long as the weight correlation matrix is a diagonal full-rank matrix, which is often found in the distributed setting. This is a consequence from the signals being independent of each other. This will have additional consequences into the estimator construction, in particular it will imply construction errors for the estimators of $\boldsymbol{\beta}$, as before. We discuss all of this in the next section, as well as provide a necessary and sufficient condition for the pseudo-invertibility.

\section{Invertibility, Pseudo-Invertibility and its consequences}\label{section:inverting}

In this section we condense and analyze the conditions for the invertibility and pseudo-invertibility of the alternant matrix and the design matrix (and implicitly once more for the Vandermonde matrix). We also go through the meaning of failing to be able to invert them. In the end, this is directly connected with wether the points at which we interrogate the underlying field allow the different functions of our model to be discerned or not.

We can find all the solutions in all the problems by finding the correspondent matrix and dealing with its invertibility. Let $\mathcal{F}$ and $X \subset D \subseteq \mathbb{K}^m$ be a set of $k$ functions on the domain $D$, and a discrete set of $p$ points in the same domain of the functions, respectively:

\begin{enumerate}[label=(\roman*)]
	\item \textit{Vandermonde Matrix}: $p = k$, $\mathcal{F} \in \mathbb{R}[x_1, x_2, \dots, x_m]$ (polynomial functions)
	\item \textit{Alternant Matrix}: $p = k $, $\mathcal{F} \in C^\infty(D)$ (analytical functions)
	\item \textit{Design Matrix}: $p \geq k$, $\mathcal{F} \in \mathcal{O}_D$ (holomorphic functions)
\end{enumerate}
From the, the inclusion presented in the beginning in Eq.~\ref{eq:inclusion} becomes obvious. Let us state the necessary and sufficient condition for the pseudo-invertibility of the design matrix:

\begin{theorem}
Let $\mathbb{K}$ be a field. Let $U$ be a subset of $\mathbb{K}^m$. Let $X$ be a discrete subset of $p$ distinct elements in $U$, $X\subset U$. Let $\mathcal{F} = \{ f_1, f_2,\cdots, f_k\}$ be a family of holomorphic functions in $U$, $f_j : U \subseteq \mathbb{K}^m \rightarrow \mathbb{K}$, that live in the ring $\mathcal{O}_U$. Let $\mathsf{span}_{\mathbb{K}} \ \mathcal{F}$ be the vector space generated by the functions in $\mathcal{F}$. Then:
\begin{equation}
	\mathcal{O}_U \supseteq I(X) \cap \mathsf{span}\  \mathcal{F} = \emptyset \Leftrightarrow \rank \ \boldsymbol{X}(X,\mathcal{F}) = k,
\end{equation} 
where $(\boldsymbol{X})_{ij} = f_j(x_i)$.
\end{theorem}

\begin{proof}
The proof follows very similarly to the proof of Thm.~\ref{thm:informal}. Suppose \(I(X) \cap \mathsf{span} \ \mathcal{F} = \emptyset\). This means that no non-zero linear combination of the functions \(f_1, f_2, \ldots, f_k\) in \(\mathsf{span} \ \mathcal{F}\) vanishes on all points in \(X\). Consequently, the evaluations of these functions at the points in \(X\) must be linearly independent, ensuring that \(\boldsymbol{X}(X, \mathcal{F})\) has full rank \(k\).

Conversely, suppose \(\rank \ \boldsymbol{X}(X, \mathcal{F}) = k\). This implies that the columns of the matrix (evaluations of the functions in \(\mathcal{F}\) at the points in \(X\)) are linearly independent. Hence, no non-zero linear combination of these functions can be zero on all points in \(X\), meaning \(I(X) \cap \mathsf{span} \ \mathcal{F} = \emptyset\).
\end{proof}

Left to discuss is the case when the rank of the design matrix is smaller than $k$, the number of functions or features one is trying to access. To introduce this discussion, let us give the simplest example:

\begin{example}
Suppose a family of functions $\mathcal{F} = \{f_1, f_2, \dots, f_k \}$ and a discrete set of points in the domain of the functions $X=\{x_1,x_2,\dots,x_p\}$ such that $\exists j : f_j (x) = 0,  \forall \ x \in X$. Assume this is the only linear combination in $\mathsf{span}\ \mathcal{F}$ that belongs to $I(X)$. Then, one can simply solve the problem exactly for $\mathcal{F}' = \{f_1, f_2, \dots, f_k \} \setminus {f_j}$, which has full rank $k-1$, and is able to find every field coefficient $\beta_i$ except for $\beta_j$.
\end{example}

However:
\begin{example}
Suppose a family of functions $\mathcal{F} = \{f_1, f_2, \dots, f_k \}$ and a discrete set of points in the domain of the functions $X=\{x_1,x_2,\dots,x_p\}$ such that $\exists j,l : f_j (x) + f_l (x) = 0,  \forall \ x \in X$. Assume this is the only linear combination in $\mathsf{span}\ \mathcal{F}$ that belongs to $I(X)$. We can never find the estimator for $\beta_j$ that does not contain an error proportional to $\beta_l$. To understand this create either the system of equations for $\mathcal{F}_j = \{f_1, f_2, \dots, f_k \} \setminus \{f_j\}$ and $\mathcal{F}_l = \{f_1, f_2, \dots, f_k \} \setminus \{ f_l\}$ and verify that each corresponds to a full rank system of $k-1$ equations, differing from the original problem by a term proportional to $\beta_l$ and $\beta_j$, respectively. 
\end{example}

From these two examples, the consequence should be apparent. One can only recover a linear subspace of coefficients $\beta_j$ that is the same size as the rank of the design matrix. Given that its maximum rank is the number of functions, which is the same as the number of coefficients, one could recover estimators for every $\beta_j$. If this is no longer true, then one can no longer get the same information as before, and the available information might have an error. Let us show what this error is in the following statement.

\begin{theorem}\label{thm:noninvertibility}
Let $\mathcal{F} = \{ f_1, f_2,\cdots, f_k\}$ and $X=\{ x_1, x_2,\dots,x_p\}$ be our regular set of holomorphic functions over some domain and a discrete set of points over that domain, respectively. Let $\mathcal{N}$ be the basis for the set of vectors in $\mathsf{span} \ \mathcal{F}$ that belong to $I(X)$. 
\begin{enumerate}[label=(\roman*)]
	\item $\rank \boldsymbol{X}(X,\mathcal{F}) = k - \mathsf{dim} \ \mathcal{N} = k'$
	\item One can build non-zero $ \boldsymbol{X}' = \boldsymbol{X}(X,\mathcal{F}') $, where $\mathcal{F} = \mathcal{F}' \cup  \mathcal{F}_0 $ and $\mathcal{F}'$ has dimension $k'$, that are full rank. Moreover, the choice is not unique
	\item Each $\beta_j$ correspondent to a function in $ \mathcal{F}'$ will have an error proportional to the $\beta_l$s correspondent to functions in $ \mathcal{F}_0$
\end{enumerate}
\end{theorem}

\begin{proof}
$[(i)]$ The proof is trivial as $\mathcal{N}$ is the null space of $\boldsymbol{X}$.
$[(ii),(iii)]$ This is a consequence of having combinations of the coefficients $\beta_j$ that cannot be found. In particular, the vectors $\boldsymbol{n}\in\mathcal{N}$ define combinations for the $\beta$s that are not available. This means that we are only able to find combinations of $\beta$s without any error that are orthogonal to all vectors in $\mathcal{N}$. This set is the size of $k'$, and can be used to reconstruct multiple subspaces of $\beta$s of the same size, with errors in an orthogonal subspace. 
\end{proof}

This implies that the non-invertibility of the design matrix provides the error of our estimators for the coefficients. The theorem above implies that there is a deep connection between what we are trying to estimate, and how we place our sensors. As a consequence of Thm.~\ref{thm:noninvertibility} one can also extract the following:

\begin{corollary}[Error-free Subspace]\label{thm:errorfreesubspace}
Let again $\mathcal{F} = \{ f_1, f_2,\cdots, f_k\}$ and $X=\{ x_1, x_2,\dots,x_p\}$ be our regular set of holomorphic functions over some domain, and a discrete set of points over that domain, respectively. Let $\mathcal{N}$ be the basis for the set of vectors in $\mathsf{span} \ \mathcal{F}$ that belong to $I(X)$, and $\mathcal{N}^\perp$ its orthogonal subspace. Any combination $\boldsymbol{n}^\perp\cdot \boldsymbol{\beta}$ for $\boldsymbol{n}^\perp \in \mathcal{N}^\perp$ can be estimated without a construction error.
\end{corollary}

This extends the results in \cite{Sekatski2019,Qian2019,Qian2020,Bringewatt2021,Hamann2023}, as we provide a comprehensive analysis about wether or not one is able to find the ``orthogonal subspace'', and if its construction has any error. This has an interesting consequence: knowledge of the model allows one to, by choosing specific positions, ignore some signals, and potentially require less sensors. This can be done by placing the sensors in positions for which the design matrix is not full-rank for the given model. However, the target property, represented by the vector $\boldsymbol{n}^\perp$ must be orthogonal to the null-space of the design matrix.


\subsection{Example: Gravitic Field Invariant Sensing}

As an example of a problem where we take advantage of Corol.~\ref{thm:errorfreesubspace}, let us present a situation where we have information about a mass from an object, and place the sensors in a way such that the estimation is oblivious to that mass, and able to retrieve information about other surrounding masses:

\begin{problem}
Suppose we have 6 sensors (blue dots) along a circle with radius $R$, $X = \{ x_1, x_2, x_3, x_4, x_5 \}$. There exists a mass along a line crossing the center of the circle, at depth $z_1$ (black dot), and three masses along a circle with radius $r$ at a depth of $z_1$, at positions $Z =  \{z_2, z_3, z_4\}$ (red dots). This corresponds to a model of the field given by:
\[
	\hat{F}(\boldsymbol{x}) = \sum_{i=1}^4 \beta_i \cdot \frac{1}{\norm{\boldsymbol{x}-\boldsymbol{z}_i}^2} + \beta_0
\]
\begin{center}
\begin{tikzpicture}[line cap=round,line join=round,scale=2,3d view={120}{30}]
\draw[-latex] (0,0,0) -- (1.5,0,0);
\draw[-latex] (0,0,0) -- (0,1.5,0);
\draw[-latex] (0,0,0) -- (0,0,-1.5);

\mycircle{0}{0}{1}{1}{0}
\mycircle{0}{0}{1}{0.5}{-1}

\fill[black] (0,0,-1) circle(0.05); 

\fill[red] (0.25,0.4330127019,-1) circle(0.05); 
\fill[red] (-0.5,0,-1) circle(0.05); 
\fill[red] (0.25,-0.4330127019,-1) circle(0.05); 

\foreach \i in {0,1,...,5} {
    \pgfmathsetmacro{\theta}{360*\i/6} 
    \pgfmathsetmacro\x{cos(\theta)} 
    \pgfmathsetmacro\y{sin(\theta)} 
    \pgfmathsetmacro\z{0} 
    \fill[blue] (\x,\y,\z) circle(0.05); 
}

\end{tikzpicture}
\end{center}

Given the sensor positions, and local evaluation of the field $\boldsymbol{F} = \{F(x)\}_{x \in X}$, find the coefficent of the source at $z_2$.
\end{problem}

Using the methodology introduced in this section, the solution of the problem comes from finding the estimator for the value $\beta_2$. Moreover, since all sensors lie in a region for which the source at $z_1$ provides the same signal strength, then they cannot estimate it, albeit being able to discern all the other values. The solution comes from building the design matrix and then the vector $\boldsymbol{c}$ in the following way:
\begin{equation}
	\beta_1 = \left[ (\boldsymbol{X}^T)^{+} \boldsymbol{e}_2 \right] \cdot \boldsymbol{F} = \boldsymbol{c} \cdot \boldsymbol{F}, \ \  (X)_{ij} =  \frac{1}{\norm{\boldsymbol{x}_i-\boldsymbol{z}_j}^2}
\end{equation}
Note that there are multiple pseudo-inverses, but the target vector $\boldsymbol{e}_2$ is orthogonal to the subspace with error, which is generated by $\boldsymbol{e}_0$ and $\boldsymbol{e}_1$, corresponding to the constant signal coming from the constant background $\beta_0$ and the signal of the mass in $z_1$.

\section{Discussion and Outlook}

As stated in the beginning of this work, we have systematically described a wide-range of problems for which of entangling a network of quantum sensors would provide a fundamental advantage. This advantage comes directly from the non-locality of what a quantum strategy is capable of obtaining. To be able to describe an estimator for a property of a field by a linear combination of a set of local values that are accessible by quantum sensors is the key point of this work. In doing so, one creates a class of problems for which quantum advantage appears naturally. In particular, the set of quantum strategies that work particularly well, are know to be achievable when the dynamics ruling the encoding are single-qubit dynamics. In the latter case, there is a set of commuting measurements that attain the Crámer-Rao bound and, therefore, make it suitable to apply. 

On the other hand, our description clarifies the consequences of the position of the sensors and how it impacts the estimators errors, which are shown to be tied to the models of our fields. Among our findings lie the error-free subspaces of the parameters, which we prove to have no error of construction. This means that their expected value has no error other than what the model of the field is able to capture. In doing so, one can also reduce the requirements for the number of sensors. Using this framework, one can also think of finding optimal placement of the sensors for a given model, which we leave for future work. 

By employing the tools we have used, we also allow, not only to capture a wider set of problems that overlaps with multi-party quantum computation, and even quantum machine learning, but we also use the same tools to discuss additional questions. These additional questions are linked with the field of hypothesis testing and modeling the real-world, for which there is plenty of literature in classical theory. We believe that this formulation helps to connect with classical estimation problems, especially in understanding which questions about the field one is able to answer given access to sets of local data.

Hard questions, such as what positioning of the sensors better allows parameters of the model to be measured, can now be solved by a comprehensive analysis over concepts in analytical geometry. For this, the literature in \cite{FOAG} can be of great help, and the design of experiments for quantum distributed sensing could be further simplified.

\section*{Acknowledgements}

The authors acknowledge the support from the EU Quantum Flagship project QIA (101102140) and France 2030 under the French National Research Agency projects HQI ANR-22-PNCQ-0002 and the PEPR integrated project EPiQ ANR-22-PETQ-0007. L.B. and Y.O. thank the support from Funda\c c\~ao para a Ci\^encia e a Tecnologia (FCT, Portugal), namely through project UIDB/04540/2020 and contract LA/P/0095/2020. L.B. acknowledges the support of FCT through scholarship BD/05268/2021.

\appendix

\section{Invertibility
of Vandermonde Matrix \label{appendix:vandermonde}}

In this appendix we pick up on the informal theorem (Thm.~\ref{thm:informal}) in the main text to prove a more formal version of it, using tools of algebraic geometry, for which we give a brief introduction focusing on the aspects we will use in the Appendix~\ref{appendix:algebraic}. Let us start by stating a condition:

\begin{condition}\label{condition:invertibility1}
Let $X$ be a subset in $\mathbb{K}^m$ and $L=\{ p_1, p_2, \cdots , p_n \}$ a set of polynomials in $m$ variables. Denote its span $\mathsf{span} L$, which is the vector space generated from taking linear combinations of polynomials in $L$. Let 
\[ 
I(X)=\langle f_1,f_2, \cdots, f_k  \rangle = \left\{ \sum_{i=1}^k r_i f_i \  \Big| \ r_i \in \mathbb{K}[x_1,x_2,\cdots,x_m] \right\}
\]
 be the ideal of $X$ in the polynomial ring $\mathbb{K}[x_1,x_2,\cdots,x_m]$. The condition can be present then as:
\begin{equation*}
	\mathsf{span} L \cap I(X) = \emptyset .
\end{equation*}
\end{condition}

Using this condition, we can state a theorem on the invertibility of the matrix.  We will state the theorem with respect to the rank of the Vandermonde matrix, which implies its invertibility. Later on, this will give us straightforwardly new results under different assumptions when the matrices are no longer square.

\begin{theorem}\label{thm:invertibility1}
Let $X$ be a discrete subset of points in $\mathbb{K}^m$ containing $|X|=p$ different elements, and $L$ a set of monomials of the same size. Let $V(X,L)$ be the Vandermonde matrix taking points in $X$ and functions in $L$.
\begin{equation*}
	\text{Condition~\ref{condition:invertibility1} is verified for $X,L$ } \Leftrightarrow \mathsf{rank} \ V(X,L) = p
\end{equation*}
\end{theorem}

\begin{proof}
The proof is straightforward when taking the definition of ideals and the function $I(\cdot)$ that takes an affine subspace in $\mathbb{K}^m$ to the polynomials in the polynomial ring that vanish at that subspace. In particular, if the rank of $V(X,L)$ is $p$ is equivalent to proving that all lines (or columns as $V$ is square) are linearly independent. In fact, from the condition stated we get that there is no vector $\boldsymbol{a}$ such that $\sum_{i=1}^{n} a_i p_i \in I(X)$. In particular, this means $\sum_{i=1}^{n} a_i p_i (x) = 0 \forall x \in X$ is never true, which is equivalent as saying that no linear combination of the columns of the Vandermonde matrix sums to 0, meaning, they are all linearly independent.
\end{proof}

Using this, one can look at the ideal of a set of points and directly make a statement about which set of derivatives $L$ one can find, $i.e.$ which system is invertible, guaranteeing an error of interpolation in the order of the derivatives with corresponding degrees in the cover of $L$, $\mathsf{Cov}(L)$. Moreover, the converse statement is also true: if we find a non-invertible matrix for a set of points, then those points belong to a variety spanned by the set of polynomials that make up the Vandermonde matrix. In Appendix~\ref{appendix:algebraic} we give some additional tools to understand the examples below.

Let us provide two important examples for possible configurations of sensors that generate invertible Vandermonde matrices, which we will then use together with a set of transformations on the positions that preserves the invertibility. Note that in these examples $V(\cdot)$ stands for the zero-locus of an set of polynomials, according to what is defined in Appendix~\ref{appendix:algebraic}, and not to the Vandermonde matrix.

\begin{example}
Let $X$ be the set of points across an equally spaced grid with $m$ points in each side, in 2 dimensions:
\begin{equation}
	X = \{  (  a_i,  b_i ), \text{ for }  a_i, b_i = 0, 1, \cdots , m-1 \} =  L.
\end{equation}
Let $I(X) = \langle (x-a_1)(x-a_2)\cdots (x-a_m) , (y-a_1)(y-a_2)\cdots (y-a_m) \rangle = \langle I_1, I_2 \rangle$ be the corresponding generating minimal ideal. The minimal simply restricts the generator set the the minimal set, meaning if we remove any, we do not get $X$.
\begin{center}
\begin{tikzpicture}[scale=0.5]
\draw[help lines, color=gray!30, dashed] (0,0) grid (4.9,4.9);
\draw[->,ultra thick] (0,0)--(5,0) node[right]{$x$};
\draw[->,ultra thick] (0,0)--(0,5) node[above]{$y$};
\foreach \x in {1,2,3,4,5 }{
   \foreach \y in {1,2,3,4,5}{
	\node[orange] at (\x,\y) [circle,fill,inner sep=1.pt]{};
    }
}
\draw[-,thick] (-1,2.5)--(-1,2.5) node[above]{$X = $};
\draw[-,thick] (6,2.5)--(6,2.5) node[above]{$ = $};
\end{tikzpicture} 
\begin{tikzpicture}[scale=0.5]
\draw[help lines, color=gray!30, dashed] (0,0) grid (4.9,4.9);
\draw[->,ultra thick] (0,0)--(5,0) node[right]{$x$};
\draw[->,ultra thick] (0,0)--(0,5) node[above]{$y$};
\draw[-,thick,orange] (1,0)--(1,5) node[above]{$a_1$};
\draw[-,thick,orange] (2,0)--(2,5) node[above]{$a_2$};
\draw[-,thick,orange] (3,0)--(3,5) node[above]{$a_3$};
\draw[-,thick,orange] (4,0)--(4,5) node[above]{$a_4$};
\draw[-,thick,orange] (5,0)--(5,5) node[above]{$a_5$};
\node at (2.5,6.5) {$V(I_1)=X_1$};
\end{tikzpicture} 
\hspace{-0.5cm}
\begin{tikzpicture}[scale=0.5]
\draw[help lines, color=gray!30, dashed] (0,0) grid (4.9,4.9);
\draw[->,ultra thick] (0,0)--(5,0) node[right]{$x$};
\draw[->,ultra thick] (0,0)--(0,5) node[above]{$y$};
\draw[-,thick,orange] (0,1)--(5,1) node[above]{$b_1$};
\draw[-,thick,orange] (0,2)--(5,2) node[above]{$b_2$};
\draw[-,thick,orange] (0,3)--(5,3) node[above]{$b_3$};
\draw[-,thick,orange] (0,4)--(5,4) node[above]{$b_4$};
\draw[-,thick,orange] (0,5)--(5,5) node[above]{$b_5$};
\draw[-,thick] (-1,2.5)--(-1,2.5) node[above]{$\bigcap $};
\node at (2.5,6.5) {$V(I_2)=X_2$};
\end{tikzpicture}
\end{center}

This means that the first variety that holds all points of the grid has \footnote{We used the notation $deg_{x_i}$ to denote the maximum degree of the coordinate $x_i$ in the polinomial corresponding to the variety.}$deg_x = deg_y = m$. In particular this means that we can always guarantee invertibility if we choose monomials $x^\alpha y^\beta$ such that $0 \leq \alpha ,\beta < m$, as they can never span a polynomial with the minimum degree of ideal of $X$. This is exactly choosing the lower set $X \subset (\mathbb{N}_0^2,\preceq)$ for the exponents of the monomials.

\end{example}

The next step is taking a more complex lower set, $i.e.$ one with a border (or Pareto front) larger than 1.

\begin{example}\label{example:borderideals}
Let $X$ be a lower set of $(\mathbb{N}_0^2,\preceq)$ with the regular product order. Let the border of $X$, $\partial X$, (or Pareto front) be given by the set of points $\partial X = \{ (A_i, B_i) \}$. Without loss of generality let $\partial X$ be ordered by the lexicographic order, ascending in the first coordinate. Consider the following family of subsets with corresponding ideals:

\begin{center}
\begin{tikzpicture}[scale=0.5]
\draw[help lines, color=gray!30, dashed] (0,0) grid (6.9,6.9);
\draw[->,ultra thick] (0,0)--(7,0) node[right]{$x$};
\draw[->,ultra thick] (0,0)--(0,7) node[above]{$y$};
\draw[-,thick,orange] (4,2)--(4,2) node[above]{$\cdots$};
\draw[-,thick,orange] (2,3.5)--(2,3.5) node[above]{$\vdots$};
\foreach \x in {1,2,4,5,6 }{
   \foreach \y in {1,...,\x}{
   	\ifnum \y = 4
		
	\else
		\node[orange] at (7-\x,\y) [circle,fill,inner sep=1.pt]{};
	\fi
    }
}
\node[orange] at (6,2) [circle,fill,inner sep=1.5pt]{};
\node[orange,above,right] at (6,2) {\scriptsize$A_1,B_1$};
\node[orange] at (5,3) [circle,fill,inner sep=1.5pt]{};
\node[orange,above,right] at (5,3) {\scriptsize$A_2,B_2$};
\node[orange] at (2,6) [circle,fill,inner sep=1.5pt]{};
\node[orange,above,right] at (2,6) {\scriptsize$A_p,B_p$};
\node[orange] at (3,5) [circle,fill,inner sep=1.0pt]{};
\draw[-,thick] (-1,3)--(-1,3) node[above]{$X = $};
\draw[-,thick] (8,3)--(8,3) node[above]{$= $};
\end{tikzpicture} 
\begin{tikzpicture}[scale=0.5]
\draw[help lines, color=gray!30, dashed] (0,0) grid (6.9,6.9);
\draw[->,ultra thick] (0,0)--(7,0) node[right]{$x$};
\draw[->,ultra thick] (0,0)--(0,7) node[above]{$y$};
\draw[-,thick,orange] (1,0)--(1,7) node[above]{$a_1$};
\draw[-,thick,orange] (2,0)--(2,7) node[above]{$a_2$};
\draw[-,thick,orange] (3,0)--(3,7) node[above]{$a_3$};
\draw[-,thick,orange] (4,3.2)--(4,3.2) node[above]{$\cdots$};
\draw[-,thick,orange] (5,0)--(5,7) node[above]{$\scalebox{0.9}[1]{\medmuskip=-2mu$A_1\shortminus1$}$};
\draw[-,thick,orange] (6,0)--(6,7) node[above]{$A_1$};
\node at (3.5,8.5) {$V(I_1)=X_1$};
\end{tikzpicture} 
\hspace{-0.5cm}
\begin{tikzpicture}[scale=0.5]
\draw[help lines, color=gray!30, dashed] (0,0) grid (6.9,6.9);
\draw[->,ultra thick] (0,0)--(7,0) node[right]{$x$};
\draw[->,ultra thick] (0,0)--(0,7) node[above]{$y$};
\draw[-,thick,orange] (1,0)--(1,7) node[above]{$a_1$};
\draw[-,thick,orange] (2,0)--(2,7) node[above]{$a_2$};
\draw[-,thick,orange] (3,0)--(3,7) node[above]{$a_3$};
\draw[-,thick,orange] (4,3.2)--(4,3.2) node[above]{$\cdots$};
\draw[-,thick,orange] (5,0)--(5,7) node[above]{$A_2$};
\draw[-,thick,orange] (0,1)--(7,1) node[above]{$B_1-1$};
\draw[-,thick,orange] (0,2)--(7,2) node[above]{$B_1$};
\draw[-,thick] (-1,3)--(-1,3) node[above]{$\bigcap $};
\node at (3.5,8.5) {$V(I_2)=X_2$};
\end{tikzpicture}

\hspace{3.5cm}
\begin{tikzpicture}[scale=0.5]
\draw[help lines, color=gray!30, dashed] (0,0) grid (6.9,6.9);
\draw[->,ultra thick] (0,0)--(7,0) node[right]{$x$};
\draw[->,ultra thick] (0,0)--(0,7) node[above]{$y$};
\draw[-,thick,orange] (1,0)--(1,7) node[above]{$a_1$};
\draw[-,thick,orange] (2,0)--(2,7) node[above]{$A_p$};
\draw[-,thick,orange] (0,3)--(7,3) node[above]{$b_3$};
\draw[-,thick,orange] (4,3.2)--(4,3.2) node[above]{$\cdots$};
\draw[-,thick,orange] (0,5)--(7,5) node[above]{$B_{p-1}$};
\draw[-,thick,orange] (0,2)--(7,2) node[above]{$b_2$};
\draw[-,thick,orange] (0,1)--(7,1) node[above]{$b_1$};
\draw[-,thick] (-2,3)--(-2,3) node[above]{$(\cdots) \quad \bigcap $};
\node at (3.5,8.5) {$V(I_p)=X_p$};
\end{tikzpicture} 
\hspace{-0.5cm} 
\begin{tikzpicture}[scale=0.5]
\draw[help lines, color=gray!30, dashed] (0,0) grid (6.9,6.9);
\draw[->,ultra thick] (0,0)--(7,0) node[right]{$x$};
\draw[->,ultra thick] (0,0)--(0,7) node[above]{$y$};
\draw[-,thick,orange] (0,1)--(7,1) node[above]{$b_1$};
\draw[-,thick,orange] (0,2)--(7,2) node[above]{$b_2$};
\draw[-,thick,orange] (0,3)--(7,3) node[above]{$b_3$};
\draw[-,thick,orange] (3.2,4)--(3.2,4) node[above]{$\vdots$};
\draw[-,thick,orange] (0,5)--(7,5) node[above]{$\scalebox{0.9}[1]{\medmuskip=-2mu$b_{m-1}$}$};
\draw[-,thick,orange] (0,6)--(7,6) node[above]{$B_p$};
\draw[-,thick] (-1,3)--(-1,3) node[above]{$\bigcap $};
\node at (3.5,8.5) {$V(I_{p+1})=X_{p+1}$};
\end{tikzpicture}

\end{center}

Where we define $X_\alpha $, $\alpha \in 1,2,\cdots , |\partial X|+1 = p+1$ according to the successive subsets above. Each of them has corresponding ideal given by:
\[
 I_\alpha = \prod_{a_i = a_1}^{A_\alpha} (x-a_i) \prod_{b_i = b_1}^{B_{\alpha-1}} (y-b_i).
\]

We get that $ V(I) = V(\bigcup_\alpha I_\alpha ) = \bigcap_\alpha V(I_\alpha) = \bigcap_\alpha X_\alpha = X$. Note again that the generating set for the ideal $I$ is minimal in the sense that if we remove any, we no longer get $X$. This means that the varieties that hold all points of the grid have $(deg_x, deg_y) \in \mathsf{Cov}(L)$. This implies we can always guarantee invertibility if we choose monomials $x^\alpha y^\beta$ such that $( \alpha ,\beta )  \in L$, as they are covered by $\mathsf{Cov}(L)$. 
\end{example}

\section{Algebraic Geometry - important concepts \label{appendix:algebraic}}

In the main text, to verify that lower sets on the space of sensors, with accompanying lower sets on the space of monomials result in a invertible Vandermonde matrix, we resort to tools from algebraic geometry to prove the space of varieties containing the points of a lattice that form a lower set implies a lower set structure on the set of monomials. Let us give a very brief set of important definitions and theorems in algebraic geometry to provide some context to our methods. The intuition behind the use of algebraic geometry is so that we are able characterize subspaces of the affine space, as the zeros of subsets of polynomial function. We refer to the following set of books and notes for further inspection \cite{FOAG,Milne2023,Fulton2008}.

From algebraic theory we get that a set $X$ in an affine space $\mathbb{A}^m$ can be seen as the zero set of an ideal in the polynomial ring $\mathbb{k}\left[x_1, \ldots, x_n\right]$. This is exactly the definition of an algebraic set:

\begin{equation}
\mathrm{V}(S)=\left\{p \in \mathbb{A}^m \mid f(p)=0 \text { for all } f \in S\right\}=\bigcap_{f \in S} \mathrm{~V}(f)
\end{equation}

While $S$ in general does not need to be an ideal, one can always find an ideal $I$ which is generated by $S$, such that $V(S) = V(\langle S \rangle) = V(I)$. Here, ideals are defined over the ring in question, in this case the polinomial ring over some algebraically closed field, and has similar properties as an order ideal:

\begin{definition}[Ideal]
Let $(R,+,\cdot)$ be a ring and $(R,+)$ its additive group. We call ideal to a subset $I$ such that:
\begin{enumerate}
	\item $(I,+)$ is a sub-group of $(R,+)$
	\item $\forall \ x \in I , y \in R$ implies $y x \in I$ 
\end{enumerate}
\end{definition}
This definition is technically for right-side ideals, with similar ones for left-side ideals. If a ideal is both right and left sided, than it is two-sided ideal, which from now on we will always assume they are, as to work with $\mathbb{R}$ and $\mathbb{C}$ is always the case, or any commutative ring.

There is as well a function that does in inverse statement, $i.e.$ pulls back an ideal from the set of points $X$ in the affine-space $\mathbb{A}^m$:

\begin{equation}
\mathrm{I}(X)=\left\{f \in \mathbb{k}\left[x_1, \ldots, x_n\right] \mid f(p)=0 \text { for all } p \in X\right\}
\end{equation}

With this, there comes a bunch of interesting properties, which one can verify over any of the suggested books \cite{FOAG,Milne2023,Fulton2008}:

\begin{enumerate}
	\item If $ S \subseteq T \subseteq \mathbb{k}\left[x_1, \ldots, x_n\right] \text { then } \mathrm{V}(T) \subseteq \mathrm{V}(S)$
	\item If $\{ I_\alpha \}$ is a collection of ideals, then $\mathrm{V}\left(\bigcup_\alpha I_\alpha\right)=\bigcap_\alpha \mathrm{V}\left(I_\alpha\right)$
	\item $V(IJ) = V(I) \cup V(J)$
	\item $\mathrm{V}(0)=\mathbb{A}^m, \mathrm{~V}(1)=\varnothing$, and $\mathrm{V}\left(x_1-a_1, \ldots, x_m-a_m\right)=\left\{\left(a_1, \ldots, a_m\right)\right\}$
	\item If $X \subseteq Y \subseteq \mathbb{A}^m$ then $\mathrm{I}(Y) \subseteq \mathrm{I}(X)$
	\item $\mathrm{I}(\varnothing)=\mathbb{k}\left[x_1, \ldots, x_n\right] $. $\mathrm{I}\left(\left\{\left(a_1, \ldots, a_n\right)\right\}\right)=\left\langle x_1-a_1, \ldots, x_m-a_m\right\rangle$ for any point $\left(a_1, \ldots, a_m\right) \in \mathbb{A}^m$. $\mathrm{I}\left(\mathbb{A}^n\right)=0$ if $\mathbb{k}$ is infinite. 
	\item $S \subseteq \mathrm{I}(\mathrm{V}(S))$ for any set of polynomials $S \subseteq \mathbb{k}\left[x_1, \ldots, x_n\right]$. $X \subseteq \mathrm{V}(\mathrm{I}(X))$ for any set of points $X \subseteq \mathbb{A}^m$.
	\item $\mathrm{V}(\mathrm{I}(\mathrm{V}(S)))=\mathrm{V}(S)$ for any set of polynomials $S \subseteq \mathbb{k}\left[x_1, \ldots, x_n\right]$. $\mathrm{I}(\mathrm{V}(\mathrm{I}(X)))=\mathrm{I}(X)$ for any set of points $X \subseteq \mathbb{A}^m$.
\end{enumerate}

Intuitively, the main idea behind finding the ideals, is that they should give the smallest polynomials under which composition into a larger polynomial (in terms of degree) can be made while maintaining the zero locus of the smallest polynomials. This means that polynomials composed from an ideal will necessarily have the same zeros as the polynomials in the ideal. Another important observation from the properties, and somewhat counterintuitive, is the overlap structure of the different sets. The union of subsets of affine spaces is contained in the intersection of their ideals, and vice-versa. And the larger the description of an ideal, the smaller the collection of points (in terms of dimension). Let us give some examples:

\begin{example}
Let $\mathbb{k} = \mathbb{R}$, $m=2$, so we are working in the plane, where lines define the hyper surface and points the dimensionless algebraic varieties. A line $x-a = 0$ is defined by the ideal $I_1 = \langle x-a \rangle $. A point $(b,c)$ is defined by the ideal $I_2=\langle x-b, y-c \rangle$. The union of a line and a point is given by the ideal $\langle (x-a) (x-b) , (x-a)(y-c) \rangle=I_1 I_2$ .
\end{example}

However, not all ideals correspond to sets of points in the affine space:

\begin{example}
Let $\mathbb{k} = \mathbb{R}$, $m=2$,again. Consider the ideal $I_1 = \langle x^2 + y^2 - 1, x-1\rangle $. $x^2 + y^2 - 1 = 0$ corresponds to a circle and $x-1 = 0 $ to a line interseting the circle at one point. The intersection of these two sets is the point $x=1,y=0$ which is defined by the ideal $I_2 = \langle x-1, y \rangle$. Note however that $I_1 \neq I_2$. The idea is that since this point has multiplicity 2, this fails. Note that $x^2+y^2-1 - (x-1)(x+1) = y^2$ meaning $\langle x^2 + y^2 - 1, x-1\rangle \subseteq \langle y^2 , x-1\rangle \subseteq \langle y , x-1\rangle$ indeed.
\end{example}

Nonetheless, doing the $I(V(I))$ for any ideal, corresponding to a set of points (closed ideal) or not, always results in a closed ideal. In particular a closed ideal is also a radical ideal, which is defined by:

\begin{equation}
\sqrt{I}=\left\{a \in R \mid a^n \in I \text { for some } n>0\right\} \text {. }
\end{equation}

This example highlights a key fact: even if we begin with an ideal that does not correspond precisely to a set of reduced points (i.e., has higher multiplicity), the operation $I(V(I))$ always produces a \emph{closed ideal}, one that corresponds exactly to a geometric variety (i.e., a set of points). In fact, the ideal $I(V(I))$ is always a \emph{radical ideal}, defined as:
\begin{equation}
\sqrt{I} = \left\{ a \in R \mid a^n \in I \text{ for some } n > 0 \right\}.
\end{equation}

Thus, we always have the chain of inclusions:
\begin{equation}
	I \subseteq I(V(I)) \subseteq \sqrt{I}.
\end{equation}

This observation leads us directly to the importance of Hilbert's Nullstellensatz, which (in its strong form) establishes a precise correspondence between radical ideals and the vanishing sets they define. Essentially, this tells us that taking the vanishing set of an ideal and then pulling back the ideal of that set always returns the radical of the original ideal: $I(V(I)) = \sqrt{I}$. This is the core idea connecting geometry (sets of points) to algebra (ideals of polynomials) that classical algebraic geometry is built on.

\section{Generalized Alternant Matrix \label{appendix:alternant}}

The alternant matrix can be thought of as generalization of the Vandermonde matrix in the space of continuously differentiable functions, which polynomials are also part of. Where before we wrote algebraic varieties, as a way to characterize spaces of points described by a set of polynomial equations, now we write analytic variety, as the spaces of points described by a set of holomorphic functions. Holomorphic functions are exactly the set of functions which are said to be continuously differentiable in a neighborhood of a point, for a set of points in the domain, such that their Taylor expansion is available locally. However, this case will be more complicated to make concrete statements, such as the ones we did for the invertibility of the Vandermonde matrix over a grid. This is due to the fact that the set of zeros of a holomorphic function is not necessarily finite, nor will the generating set of the ideal be.

Similarly to the definition of a algebraic variety being a zero locus of a set of polynomials, we define a analytic variety taking a set of holomorphic functions $\mathcal{F} = \{ f_1, f_2, \cdots , f_n \}$ such that each $f_j : U \subseteq \mathbb{k}^m \rightarrow \mathbb{k}$ for some subset $U$ of the space we are working in, in the following way:

\begin{equation}
	\mathrm{V}(\mathcal{F})=\left\{x \in U \mid f_j(x)=0 \text { for all } f_j \in \mathcal{F}\right\}
\end{equation}

Moroever, we are also able to define ideals in the ring of holomorphic functions. It is known \cite{Henriksen1952,Henriksen1953} that these ideals can be either generated by one element, or infinite elements. Nevertheless, we are able to make the same statements and provide the same conditions on the invertibility of the alternant matrix taking the analytic variety instead of algebraic geometry. Namely, taking the alternant matrix over two equally sized sets of points and holomorphic functions $V(X,\mathcal{F})$, then all the points in $X$ cannot belong to some variety spanned by the functions in $\mathcal{F}$ in order to get invertibility. We will refrain from stating the theorem in this section, as in the next section we will provide a statement that is applicable to this case, but also to the generalization of this case into a linear least-squares method.

\section{Optimizing Local and Non-Local Strategies \label{appendix:optimizing}}

To optimize the local and non-local strategies under the assumption that our total resources are limited we describe a vector of resources $\boldsymbol{n}$, where each index corresponds to a the amount of resources we give to each sensor. In particular, the constraint of limited total resources comes as the 1-norm of this vector is bounded by $N$. If we express $\boldsymbol{n} = N \boldsymbol{\delta n}$, then $\norm{\boldsymbol{\delta n}}_1 = 1$, and each component can be seen as the percentage distribution of the resources. 

The target estimator is given by $\hat{H} = \boldsymbol{c}\cdot \boldsymbol{F}$. For the non-local case, by making $\boldsymbol{n} = \alpha \boldsymbol{c}$, we directly get access to $\hat{H}$ with total variance given by $1/\alpha^2$. Moreover, since $\norm{\boldsymbol{n}}_1 = N$, we get that:
\[
	\Delta \hat{H}^2 = \frac{\norm{\boldsymbol{c}}^2}{N^2}.
\]
For the local case, we have access to each of the estimators of $\hat{F}(x)$ with a precision given by $\Delta \hat{F}(x_i)^2 = 1/n_i^2$ which we then use to construct $\hat{H}$. Using error propagation we then have that:
\[
	\Delta \hat{H}^2 = \sum_{i} \frac{|c_i|^2}{n_i^2}.
\]
Note that if we have a classical strategy, we can also do the same as above, but substituting the power two in the denominator by one, has a classical case does not reach the Heisenberg scaling. Nonetheless, the problem of finding the optimal distribution of qubits can be formulated as solving the following general optimization problem:

\begin{equation}
\begin{aligned}
	\text{Minimize: }& \sum_{j=1}^n  \frac{|a_j|^q}{|b_j|^p}  \\
	\text{Subject to: }& a_1, \cdots , a_n  \in \mathbb{R} \setminus \{ 0 \} \\
	& b_1, \cdots , b_n > 0 \\
	& \sum_j b_j = N
\label{eq:optimization}
\end{aligned}
\end{equation}

One can find the solution from the following main equation:

\begin{equation}
\begin{aligned}
	 \sum_{i=1}^n  \frac{|a_i|^q}{|b_i|^p}  &= \sum_{i=1, i\neq k}^n \frac{|a_i|^q}{|b_i|^p}   +  \frac{|a_k|^q}{|N-\sum_{i\neq k} b_i|^p}  \\[10pt]
	 \frac{\partial (\cdot) }{\partial b_j} = 0 &\Longleftrightarrow   \frac{|a_j|^q}{|b_j|^{p+1}} =  \frac{|a_k|^q}{|b_k|^{p+1}}
\label{eq:optimization2}
\end{aligned}
\end{equation}
And given that one can find the same equation for every different pair of $j,k$ such that $j\neq k$, one concludes that:

\begin{equation}
\begin{aligned}
	\frac{|a_j|^q}{|b_j|^{p+1}} &= \alpha\quad , \forall j \\[10pt]
	\implies |a_j|^q &= \alpha |b_j|^{p+1} \\
	\Longleftrightarrow b_j &= \tilde{\alpha} |a_j|^{\frac{q}{p+1}} \\[10pt]
	\sum_j b_j = N &\implies \tilde{\alpha} = \frac{N}{\sum_j |a_j|^\frac{q}{p+1}} \\[10pt]
	\min_{b_1,\cdots,b_n} \sum_{j=1}^n  \frac{|a_j|^q}{|b_j|^p}  &= \frac{1}{|\tilde{\alpha}|^{p}} \sum_j |a_j|^{\frac{q}{p+1}}= \frac{1}{N^p} \left( \sum_j |a_j|^{\frac{q}{p+1}} \right)^{p+1}\\
	&= \frac{1}{N^p} \norm{\vec{a}}_{\frac{q}{p+1}}^q
\label{eq:optimization3}
\end{aligned}
\end{equation}

Where we used the notation of the norm, even though it is not really a norm, as the $p-$norm is only defined for $p\geq 1$ and $q/(p+1)$ can be smaller than 1, depending on the choice of $p,q$. We discuss the bounds of this norms over the next Appendix~\ref{appendix:norms}, as we can use them to prove that the non-local strategy is always better than the best local strategy. In particular one can also observe from the norm inequalities that only case where a vector $\boldsymbol{c}$ has one non-zero component, the non-local is equal to the local strategy. This case however, corresponds to putting every resource on one sensor only, which is in fact a local approach.

\section{$p$-Norms and $q/p-$Norms \label{appendix:norms}}

Let $x \in \mathbb{R}$ be a real vector of size $n$. It is known that the $p-$norm of a vector is well-defined as:

\begin{equation}
	\norm{x}_p = \left( \sum_{i=1}^n |x_i|^{p}\right)^{1/p}
\end{equation}

This constitutes a norm for $p\geq 1$, as it verifies all norm properties. Even more, for $p > q \geq 1$ we have that:

\begin{equation}
	\norm{x}_p \leq \norm{x}_q \leq n^{\frac{1}{q}-\frac{1}{p}} \norm{x}_p
	\label{eq:normineq}
\end{equation}

Now let us suppose $p$ and $q$ are not necessarily equal or bigger than one. This does not define a norm, as the triangle inequality no longer holds. Nonetheless we want to extend the inequality in Eq.~\ref{eq:normineq} to this case.

\begin{proposition}
Let $p,q$ be positive real numbers such that $p>q$, and $x$ be a real vector of size $n$, $x\in \mathbb{R}^n$.
\begin{equation}
	\norm{x}_p \leq \norm{x}_q \leq n^{\frac{1}{q}-\frac{1}{p}} \norm{x}_p
\end{equation}
\end{proposition}

\begin{proof}
This is already true for $p,q\geq 1$. To prove for $1>p'=p/r,q'=q/r>0$ let $r\in\mathbb{R}$ be a rescaling value such that $r>p>q\geq 1$. Then:
\begin{equation}
\begin{aligned}
	\left( \sum_{i=1}^n |x_i|^{p} \right)^{1/p} &\leq \left( \sum_{i=1}^n |x_i|^{q}\right)^{1/q}  \leq n^{\frac{1}{q}-\frac{1}{p}} \left( \sum_{i=1}^n |x_i|^{p} \right)^{1/p} \\
	\left( \sum_{i=1}^n |y_i|^{p/r}\right)^{1/p} &\leq \left( \sum_{i=1}^n |y_i|^{q/r}\right)^{1/q} \leq n^{\frac{1}{q}-\frac{1}{p}}\left( \sum_{i=1}^n |y_i|^{p/r}\right)^{1/p}\\
	\left( \sum_{i=1}^n |y_i|^{p/r}\right)^{r/p} &\leq \left( \sum_{i=1}^n |y_i|^{q/r}\right)^{r/q} \leq n^{\frac{r}{q}-\frac{r}{p}} \left( \sum_{i=1}^n |y_i|^{p/r}\right)^{r/p} \\[15pt]
	\norm{y}_{p/r} &\leq \norm{y}_{q/r} \leq n^{\frac{r}{q}-\frac{r}{p}} \norm{y}_{p/r} \qquad ,\forall y\in\mathbb{R}^n
\end{aligned}
\end{equation}
Where $|y_i| = |x_i|^r$, which is a continuous transformation, meaning if it is valid for all $x$, then we can say the same for all $y$. Now $1>p'=p/r >q' = q/r >0$, proving the left hand side inequality. The right-hand side inequality comes straightforwardly in the same way. 
\end{proof}

\section{Interpolation Error \label{appendix:interpolation}}

In this appendix we provide more insight of the accuracy, or error, of the interpolation of a polynomial function depending on the number of points in a grid on has. For a equally spaced grid of $n\times n$ sensors, one can construct an invertible Vandermonde matrix with monomials of degree $\boldsymbol{\alpha}$ such that $\norm{\boldsymbol{\alpha}}_\infty \leq n-1$. This means any polynomial containing monomials of degree within this region will have zero error. We plot in Fig.~\ref{fig:errordeg} two polynomials and their errors for the $3\times 3$ and $5 \times 5$ sensor arrays. 

\begin{figure}[ht!]
\centering
\subfloat{\includegraphics[width=0.6\linewidth]{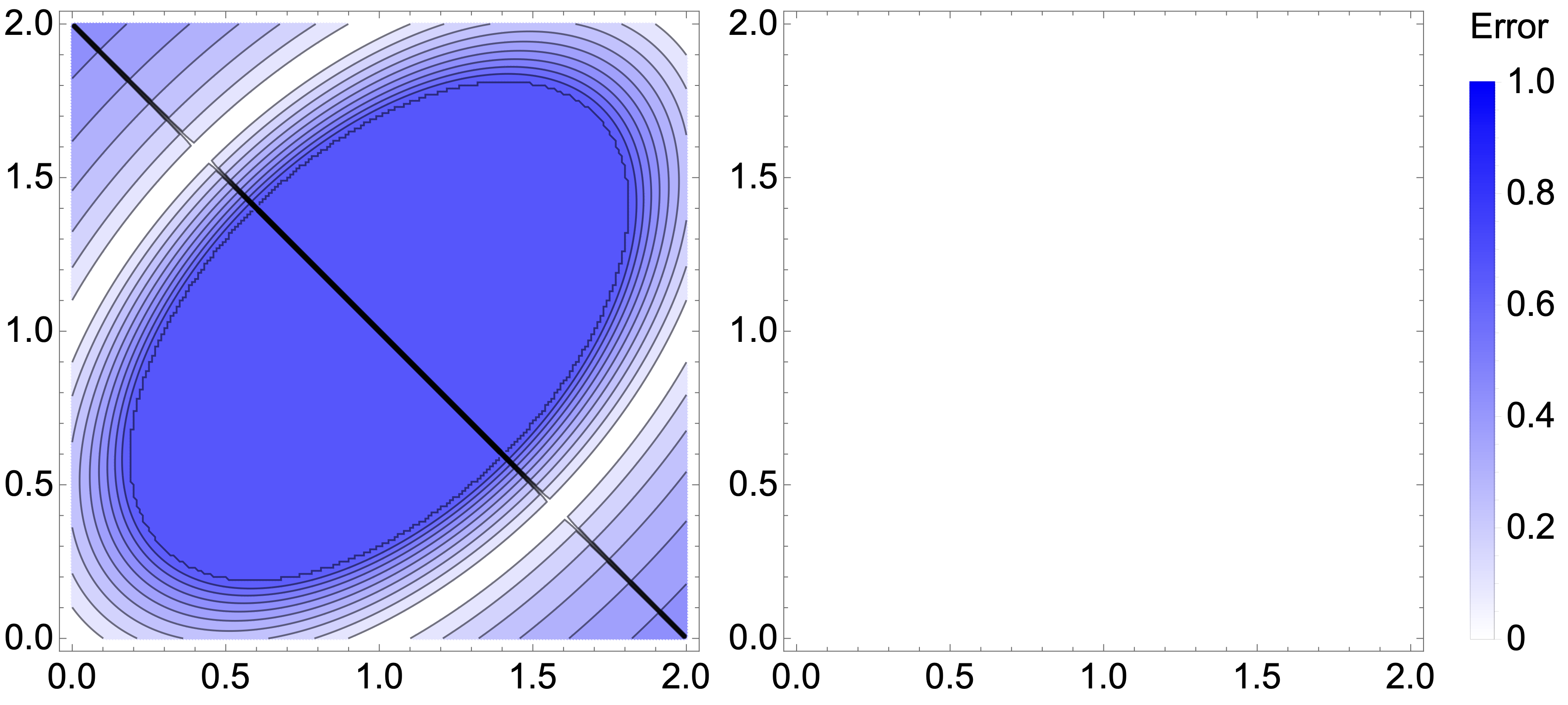}} \\
\subfloat{\includegraphics[width=0.6\linewidth]{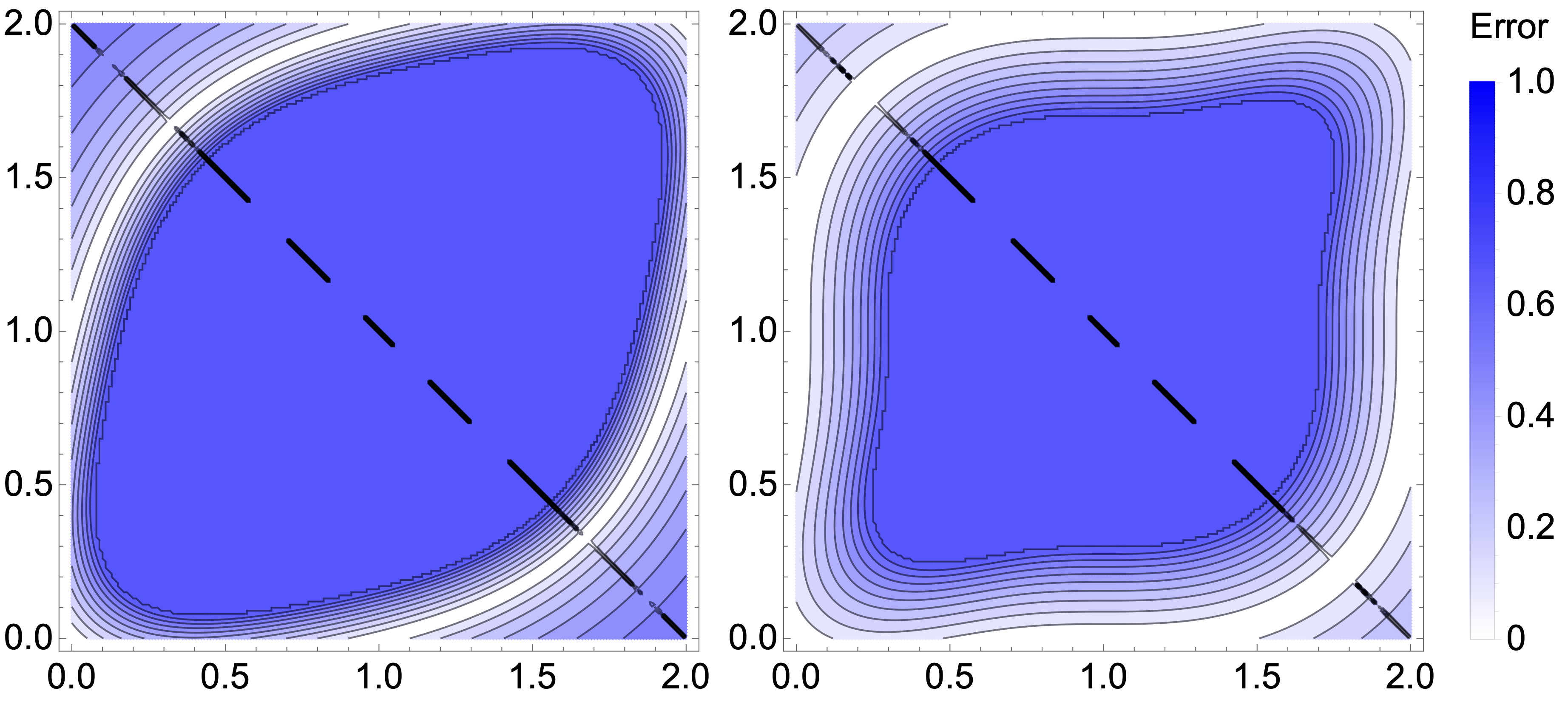}} 
\caption{\textbf{Top Row:} Comparison of the error of the interpolation of the polynomial $F(x) = (x-1)^3 + (y-1)^3$, considering the two possible equally spaced grids containing $3\times 3$ sensors (left) and $5\times 5$ sensors (right). \textbf{Bottom Row:} Comparison of the error of the interpolation of the polynomial $F(x) = (x-1)^5 + (y-1)^5$, considering the two possible equally spaced grids containing $3\times 3$ sensors (left) and $5\times 5$ sensors (right).}
\label{fig:errordeg}
\end{figure}

\end{document}